\documentclass[twocolumn]{article}
\usepackage{ltexpprt}

\usepackage{amsmath}
\usepackage{amssymb}
\usepackage{numprint}
\usepackage{hyperref}
\usepackage{graphicx}
\usepackage[title,toc,titletoc,page]{appendix}
\usepackage{flushend}
\usepackage{dblfloatfix}
\usepackage{xcolor}

\newcommand{\anoi}{\texttt{NOI}~}

\newcommand{\Oh}[1]{\mathcal{O}\!\left( #1\right)}
\newcommand{\Is}       {:=}

\newif\ifDoubleBlind
\DoubleBlindfalse

\usepackage[switch] {lineno}
\ifDoubleBlind
\linenumbers
\fi{}
\newcommand{\ie}{i.e.\ }
\newcommand{\etal}{et~al.~}
\newcommand{\eg}{e.g.\ }

\usepackage{tikz,pgfplots}
\usepgfplotslibrary{colorbrewer}
\usetikzlibrary{patterns}
\npdecimalsign{.} 

\usepackage[utf8]{inputenc}

\usepackage{algpseudocode}
\usepackage{algorithm}

\algnewcommand\algorithmicinput{\textbf{Input:}}
\algnewcommand\INPUT{\item[\algorithmicinput]}
\algnewcommand\algorithmicoutput{\textbf{Output:}}
\algnewcommand\OUTPUT{\item[\algorithmicoutput]}

\newcommand{\CC}{C\texttt{++}}

\pgfplotscreateplotcyclelist{black white 2}{%
	every mark/.append style={solid},mark=*\\%
	every mark/.append style={solid},mark=square*\\%
	every mark/.append style={solid},mark=o\\%
	mark=star\\%
	every mark/.append style={solid},mark=diamond*\\%
	densely dashed,every mark/.append style={solid},mark=*\\%
	densely dashed,every mark/.append style={solid},mark=square*\\%
	densely dashed,every mark/.append style={solid},mark=o\\%
	densely dashed,every mark/.append style={solid},mark=star\\%
	densely dashed,every mark/.append style={solid},mark=diamond*\\%
}

\pgfplotsset{
  cycle list/Dark2,
  every axis/.append style={
    ylabel near ticks,
    log basis y={2},
    log basis x={2},
    legend cell align={left},
    legend style={font=\Large},
    label style={font=\Large},
    title style={font=\Large},
    tick label style={font=\Large},
    cycle multiindex* list={
      Dark2
      \nextlist
      black white 2
      \nextlist
   },
  }
}

\def\MdR{\ensuremath{\mathbb{R}}}
\def\MdN{\ensuremath{\mathbb{N}}}

\usepackage{color}

\ifDoubleBlind
\newcommand{\mytitle}{Shared-memory Exact Minimum Cuts}
\else
\newcommand{\mytitle}{Shared-memory Exact Minimum Cuts \thanks{
    The research leading to these results has received funding from the European Research Council under the European Community's Seventh Framework Programme (FP7/2007-2013) /ERC grant agreement No. 340506}}
\fi{}
\begin{document}

\title{\Large \mytitle}
\ifDoubleBlind
\author{}
\else
\author{
  Monika Henzinger \thanks{University of Vienna, Vienna, Austria. \texttt{monika.henzinger@univie.ac.at}} \\
  \and
  Alexander Noe \thanks{University of Vienna, Vienna, Austria. \texttt{alexander.noe@univie.ac.at}}\\
  \and
  Christian Schulz\thanks{University of Vienna, Vienna, Austria. \texttt{christian.schulz@univie.ac.at}} \\
}
\fi{}

\date{}
\maketitle


\begin{abstract} The minimum cut problem for an undirected edge-weighted graph
asks us to divide its set of nodes into two blocks while minimizing the weight sum of the cut edges. In this paper, we engineer the fastest known exact algorithm for the problem.

State-of-the-art algorithms like the algorithm of Padberg and Rinaldi or the algorithm of Nagamochi, Ono and Ibaraki identify edges that can be contracted to reduce the graph size such that at least one minimum cut is maintained in the contracted graph.
Our algorithm achieves improvements in running time over these algorithms by a multitude of techniques. First, we use a recently developed fast and parallel \emph{inexact} minimum cut algorithm to obtain a better bound for the problem. Then we use  reductions that depend on this bound, to reduce the size of the graph much faster than previously possible.
We use improved data structures to further improve the running time of our algorithm.
Additionally, we parallelize the contraction routines of Nagamochi, Ono and Ibaraki.
Overall, we arrive at a system that outperforms the fastest state-of-the-art solvers for the \emph{exact} minimum cut problem significantly.
\end{abstract}
\ifDoubleBlind
\vfill
\clearpage
\setcounter{page}{1}
\else
\fi{}

\section{Introduction}

Given an undirected graph with non-negative edge weights,
the \emph{minimum cut problem} is to partition the vertices into two sets so
that the sum of edge weights between the two sets is minimized. A minimum cut is
often also referred to as the \textit{edge connectivity} of a
graph~\cite{nagamochi1992computing,henzinger2017local}. The problem has
applications in many fields. In particular, for network
reliability~\cite{karger2001randomized,ramanathan1987counting}, 
assuming equal failure probability edges, the smallest edge cut in the network has
the highest chance to disconnect the network; in VLSI
design~\cite{krishnamurthy1984improved}, a minimum cut can be used to minimize the number of
connections between microprocessor blocks; and it is further used as a subproblem in the
branch-and-cut algorithm for solving the Traveling Salesman Problem and other
combinatorial problems~\cite{padberg1991branch}.

As the minimum cut problem has many applications and is often used as a subproblem for complex problems, it is highly important to have algorithms that are able solve the problem in reasonable time on huge data sets. As data sets are growing substantially faster than processor speeds, a good way to achieve this is efficient parallelization. While there is a multitude of algorithms, which solve the minimum cut problem exactly on a single core~\cite{hao1992faster,henzinger2017local,karger1996new,nagamochi1992computing}, to the best of our knowledge, there exists only one parallel exact algorithm for the minimum cut problem: Karger and Stein~\cite{karger1996new} present a parallel variant for their random contraction algorithm~\cite{karger1996new} which computes a minimum cut with high probability in polylogarithmic time using~$n^2$ processors. This is however unfeasible for large instances. There has been a MPI implementation of this algorithm by Gianinazzi~\etal\cite{gianinazzi2018communication}.
However, there have been no parallel implementations of the algorithms of Hao~\etal\cite{hao1992faster} and Nagamochi~\etal\cite{nagamochi1992computing,nagamochi1994implementing}, which outperformed other exact algorithms by orders of magnitude~\cite{Chekuri:1997:ESM:314161.314315,henzinger2018practical,junger2000practical}, both in real-world and generated networks.

All algorithms that solve the minimum cut problem exactly have non-linear running times, currently the fastest being the deterministic algorithm of Henzinger~\etal~\cite{henzinger2017local} with running time $\mathcal{O}(\log^2 n \log \log^2 n)$. There is a linear time approximation algorithm, namely the $(2 + \varepsilon)$-approximation algorithm by Matula~\cite{matula1993linear} and a linear time heuristic minimum cut algorithm by Henzinger \etal\cite{henzinger2018practical} based on the label propagation algorithm~\cite{raghavan2007near}. The latter paper also contains a shared-memory parallel implementation of their algorithm.

\subsection{Contribution.}
We engineer the fastest known \emph{exact} minimum cut algorithm for the problem. We do so by (1) incorporating recently proposed \emph{inexact} methods and (2) by using better suited data structures and other optimizations as well as (3) parallelization.

Algorithms like the algorithm of Padberg and Rinaldi or the algorithm of Nagamochi, Ono and Ibaraki identify edges that can be contracted to reduce the graph size such that at least one minimum cut is maintained in the contracted graph.
Our algorithm achieves improvements in running time by a multitude of techniques. First, we use a recently developed fast and parallel \emph{inexact} minimum cut algorithm~\cite{henzinger2018practical} to obtain a better approximate bound $\hat\lambda$ for the problem. As know graph reduction techniques depend on this bound, the better bound enables us to apply more reductions and reduce the size of the graph much faster.
For example, edges whose incident vertices have a connectivity of at least $\hat\lambda$, can be contracted without the contraction affecting the minimum cut.
Using better suited data structures as well as incorporating observations that help to save a significantly amount of work in the contraction routine of Nagamochi, Ono and Ibaraki~\cite{nagamochi1994implementing} further reduce the running time of our algorithm. For example, we observe a significantly higher performance on some graphs when using a FIFO bucket priority queue, bounded priority queues as well as better~bounds~$\hat\lambda$. Additionally, we give a parallel variant of the contraction routines of Nagamochi, Ono and Ibaraki~\cite{nagamochi1994implementing}.
Overall, we arrive at a system that outperforms the state-of-the-art by a factor of up to $2.5$ already sequentially, and when run in parallel by a factor of up to $12.9$ using~$12$~cores.

The rest of the paper is organized as follows. Chapter \ref{prelim} gives preliminaries, an overview over related work and details of the algorithms of Nagamochi~\etal\cite{nagamochi1992computing,nagamochi1994implementing} and Henzinger~\etal\cite{henzinger2018practical}, as we make use of their results. Our shared-memory parallel exact algorithm for the minimum cut problem is detailed in Chapter~\ref{cuts}. In Chapter~\ref{impl} we give implementation details and extensive experiments both on real-world and generated graphs. We conclude the paper in Chapter~\ref{conclusion}.

\section{Preliminaries}\label{s:preliminaries}
\label{prelim}

\subsection{Basic Concepts.}
Let $G = (V, E, c)$ be a weighted undirected graph with vertex set $V$, edge set $E \subset V \times V$ and
non-negative edge weights $c: E \rightarrow \MdN$. 
We extend $c$ to a set of edges $E' \subseteq E$ by summing the weights of the edges; that is, $c(E')\Is \sum_{e=\{u,v\}\in E'}c(u,v)$. We apply the same notation for single nodes and sets of nodes.
Let $n = |V|$ be the
number of vertices and $m = |E|$ be the number of edges in $G$. The \emph{neighborhood}
$N(v)$ of a vertex $v$ is the set of vertices adjacent to $v$. The \emph{weighted degree} of a vertex is the sum of the weight of its incident edges. For brevity, we simply call this the \emph{degree} of the vertex.
For a set of vertices $A\subseteq V$, we denote by $E[A]\Is \{(u,v)\in E \mid u\in A, v\in V\setminus A\}$; that is, the set of edges in $E$ that start in $A$ and end in its complement.
A cut $(A, V
\setminus A)$ is a partitioning of the vertex set $V$ into two non-empty
\emph{partitions} $A$ and $V \setminus A$, each being called a \emph{side} of the cut. The \emph{capacity} of a cut $(A, V
\setminus A)$ is $c(A) = \sum_{(u,v) \in E[A]} c(u,v)$.
A \emph{minimum cut} is a cut $(A, V
\setminus A)$ that has smallest weight $c(A)$ among all cuts in $G$. We use $\lambda(G)$ (or simply
$\lambda$, when its meaning is clear) to denote the value of the minimum cut
over all $A \subset V$. For two vertices $s$ and $t$, we denote $\lambda(G,s,t)$ as the smallest cut of $G$, where $s$ and $t$ are on different sides of the cut. The connectivity $\lambda(G,e)$ of an edge $e=(s,t)$ is defined as $\lambda(G,s,t)$, the connectivity of its incident vertices. This is also known as the \emph{minimum s-t-cut} of the graph or the \emph{connectivity} or vertices $s$ and $t$. At any point in the execution of a minimum cut algorithm,
$\hat\lambda(G)$ (or simply $\hat\lambda$) denotes the lowest upper bound of the
minimum cut that an algorithm discovered until that point. 
For a vertex $u \in V$ with minimum vertex degree, the size of the \emph{trivial cut} $(\{u\}, V\setminus \{u\})$ is equal to the vertex degree of $u$.
Hence, the minimum vertex degree $\delta(G)$ can serve as initial~bound.

Many algorithms tackling the minimum cut problem use \emph{graph contraction}.
Given
an edge $(u, v) \in E$, we define $G/(u, v)$ to be the graph after \emph{contracting
edge} $(u, v)$. In the contracted graph, we delete vertex $v$ and all edges
incident to this vertex. For each edge $(v, w) \in E$, we add an edge $(u, w)$
with $c(u, w) = c(v, w)$ to $G$ or, if the edge already exists, we give it the edge
weight $c(u,w) + c(v,w)$.

\subsection{Related Work.}
\label{related}

We review algorithms for the global minimum cut and related problems. A closely related problem is the \textit{minimum s-t-cut} problem, which asks
for a minimum cut with nodes $s$ and $t$ in different partitions. Ford and
Fulkerson~\cite{ford1956maximal} proved that minimum $s$-$t$-cut 
is equal to maximum $s$-$t$-flow. Gomory and Hu~\cite{gomory1961multi}
observed that the (global) minimum cut can be computed with $n-1$ minimum
$s$-$t$-cut computations.
For the following decades, this result by Gomory and Hu was used to find better
algorithms for global minimum cut using improved maximum flow
algorithms~\cite{karger1996new}. One of the fastest known maximum flow
algorithms is the push-relabel algorithm~\cite{goldberg1988new} by Goldberg~and~Tarjan.

Hao and Orlin~\cite{hao1992faster} adapt the push-relabel algorithm to pass
information to future flow computations. When a push-relabel iteration is finished,
they implicitly merge the source and sink to form a new sink and find a new
source. Vertex heights are maintained over multiple iterations of push-relabel. With
these techniques they achieve a total running time of $O(mn\log{\frac{n^2}{m}})$ for
a graph with $n$ vertices and $m$ edges, which is asymptotically equal to a
single run of the push-relabel algorithm.

Padberg and Rinaldi~\cite{padberg1990efficient} give a set of heuristics for
edge contraction. Chekuri~\etal\cite{Chekuri:1997:ESM:314161.314315} give an
implementation of these heuristics that can be performed in time linear in the
graph size. Using these heuristics it is possible to sparsify a graph while
preserving at least one minimum cut in the graph. If their algorithm does not
find an edge to contract, it performs a maximum flow computation, giving the
algorithm worst case running time $O(n^4)$. However, the heuristics can also be
used to improve the expected running time of other algorithms by applying them
on interim graphs~\cite{Chekuri:1997:ESM:314161.314315}.

Nagamochi \etal\cite{nagamochi1992computing,nagamochi1994implementing} give a minimum
cut algorithm which does not use any flow computations. Instead, their
algorithm uses maximum spanning forests to find a non-empty set of contractible
edges. This contraction algorithm is run until the graph is contracted into a single
node. The algorithm has a running time of $O(mn+n^2\log{n})$. Stoer and Wagner~\cite{stoer1997simple} give a simpler variant of the algorithm of
Nagamochi, Ono and Ibaraki~\cite{nagamochi1994implementing}, which has a
the same asymptotic time complexity. The performance of this algorithm on
real-world instances, however, is significantly worse than the performance of the
algorithms of Nagamochi, Ono and Ibaraki or Hao and
Orlin, as shown in experiments conducted by J\"unger \etal\cite{junger2000practical}.
Both the algorithms of Hao and Orlin, and  Nagamochi, Ono and Ibaraki achieve close to linear running time on most benchmark instances~\cite{Chekuri:1997:ESM:314161.314315,junger2000practical}.
To the best of our knowledge, there are no
parallel implementation of either algorithm. Both of the algorithms do not have a straightforward parallel~implementation.

Kawarabayashi and Thorup~\cite{kawarabayashi2015deterministic} give a
deterministic near-linear time algorithm for the minimum cut problem, which runs
in $O(m \log^{12}{n})$. Their algorithm works by growing contractible regions
using a variant of PageRank~\cite{page1999pagerank}. It was later
improved by Henzinger~\etal\cite{henzinger2017local} to run in $O(m
\log^2{n} \log \log^2 n)$ time.

Based on the algorithm of
Nagamochi, Ono and Ibaraki, Matula~\cite{matula1993linear} gives a
$(2+\varepsilon)$-approximation algorithm for the minimum cut problem. The algorithm contracts more edges than the algorithm of
Nagamochi, Ono and Ibaraki to guarantee a linear time complexity while still
guaranteeing a $(2+\varepsilon)$-approximation~factor. Karger and
Stein~\cite{karger1996new} give a randomized Monte Carlo algorithm based on random edge contractions. This algorithm returns the minimum cut with high probability and a larger~cut~otherwise. In experiments, the algorithm was often outperformed by Nagamochi~\etal and Hao and Orlin by orders of magnitude~\cite{Chekuri:1997:ESM:314161.314315,henzinger2018practical,junger2000practical}.

\subsection{Nagamochi, Ono and Ibaraki's Algorithm.}
\label{sec:noi}
\label{capforest}
We discuss the algorithm by Nagamochi, Ono and Ibaraki~\cite{nagamochi1992computing,nagamochi1994implementing}
in greater detail since our work makes use of the tools proposed by those authors.
The intuition behind the algorithm is as follows: imagine you have an unweighted graph with minimum cut  value exactly one. 
Then any spanning tree must contain at least one edge of each of the minimum cuts. 
Hence, after computing a spanning tree, every remaining edge can be contracted without losing the minimum cut.
Nagamochi, Ono and Ibaraki extend this idea to the case where the graph can have edges with positive weight as well as the case in which the minimum cut is bounded by $\hat \lambda$.
The first observation is the following: assume that you already found a cut in the current graph of size~$\hat \lambda$ and you want to find a out whether there is a cut of size~$< \hat \lambda$. Then the contraction process only needs to ensure that the contracted graph contains all cuts having a value \emph{strictly} smaller than~$\hat \lambda$. 
To do so, Nagamochi, Ono and Ibaraki build \emph{edge-disjoint maximum spanning forests} and contract all edges that are not in one of the~$\hat\lambda - 1$ first spanning forests, as those connect vertices that have connectivity at least $\hat\lambda$. Note that the edge-disjoint maximum spanning forest \emph{certifies} for any edge $e=\{u,v\}$ that is not in the forest that the minimum cut between $u$ and $v$ is at least $\hat \lambda$. Hence, the edge can be ``safely'' contracted.
As weights are integer, this guarantees that the contracted graph still contains all cuts that are \emph{strictly} smaller than~$\hat\lambda$. 
Since it would be inefficient to directly compute $\hat \lambda - 1$ edge disjoint maximum spanning trees, the authors give a modified algorithm to be able to detect contractable edges faster. 
This is done by computing a lower bound on the connectivity of the endpoints of an edge which serves as a certificate for a edge to be contractable.
The algorithm has worst case running time
$\Oh{mn+n^2\log n}$. In experimental evaluations~\cite{Chekuri:1997:ESM:314161.314315,junger2000practical,henzinger2018practical} it is one of the fastest exact minimum cut algorithms, both on real-world and generated instances.

We now dive  into more details of the algorithm. 
To find contractable edges, the algorithm uses a modified \emph{breadth-first graph traversal} (BFS) algorithm~\cite{nagamochi1992computing,nagamochi1994implementing}. 
More precisely, the algorithm starts at an arbitrary vertex.
In each step, the algorithm visits (scans) the vertex $v$ that is most strongly connected to the already visited vertices. For this purpose a priority queue $\mathcal{Q}$ is used, in which the connectivity strength of each vertex $r: V \to \MdR$ to the already discovered vertices is used as a key. 
When scanning a vertex $v$, the value $r(w)$ is kept up to date  for every unscanned neighbor $w$ of $v$ by setting \ie $r(w) := r(w) + c(e)$.
Moreover, for each such edge $e = (v,w)$, the algorithm computes a lower bound~$q(e)$ for the connectivity, \ie the smallest cut $\lambda(G, v, w)$, which places $v$
and $w$ on different sides of the cut.
More precisely, as shown by \cite{nagamochi1994implementing,nagamochi1992computing}, if the vertices are scanned in a certain order (the order used by the algorithm), then $r(w)$ is a lower bound on $\lambda(G,v,w)$.

For an edge that has connectivity $\lambda(G,v,w) \geq \hat\lambda$, we know that there is no cut smaller than $\hat\lambda$ that places $v$ and $w$ in different partitions. If an edge $e$ is not in a given cut $(A,V \setminus A)$, it can be contracted without affecting the cut. Thus, we can contract edges with connectivity at least $\hat\lambda$ without losing any cuts smaller than $\hat\lambda$. As $q(e) \leq \lambda(G,u,v)$ (lower bound), all edges with $q(e) \geq \hat\lambda$ are contracted.

Afterwards, the algorithm continues on the contracted
graph. A single iteration of the subroutine can be performed in $\Oh{m+n \log
n}$. The authors show that in each BFS run, at least one edge of the graph can
be contracted~\cite{nagamochi1992computing}.
This yields a total running time of $O(mn+n^2 \log n)$. 
However, in practice the number of iterations is typically much less than $n-1$, often it is proportional to  $\log n$.

\subsection{Inexact Shared-Memory Minimum Cuts.}

\texttt{VieCut} is a multilevel algorithm that uses a shared-memory parallel implementation of the label propagation algorithm~\cite{raghavan2007near} to find clusters with a strong intra-cluster connectivity. The algorithm then contracts these clusters as it is assumed that the minimum cut does not split a cluster, as the vertices in a cluster are strongly interconnected with each other. This contraction is followed by a linear-work shared memory run of the Padberg-Rinaldi local tests for contractible edges~\cite{padberg1990efficient}. This whole process is repeated until the graph has only a constant amount of vertices left and can be solved by the algorithm of Nagamochi~\etal\cite{nagamochi1994implementing} exactly.

While \texttt{VieCut} can not guarantee optimality or even a small approximation ratio, in practice the algorithm finds near-optimal minimum cuts, often even the exact minimum cut, very quickly and in parallel. The algorithm can be implemented to have sequential running time $\mathcal{O}(n+m)$.

\section{Fast Exact Minimum Cuts}
\label{cuts}

In this section we detail our shared-memory algorithm for the minimum cut problem that is based on the algorithms of Nagamochi~\etal\cite{nagamochi1992computing,nagamochi1994implementing} and Henzinger~\etal\cite{henzinger2018practical}. We aim to modify the algorithm of Nagamochi~\etal\cite{nagamochi1994implementing} in order to find exact minimum cuts faster and in parallel. Their algorithm uses a routine described above in Section~\ref{capforest}, called CAPFOREST in their original work, in order to compute a lower bound $q(e)$ of the connectivity $\lambda(G,u,v)$ for each~edge~$e=(u,v)$.

If the connectivity between two vertices is larger than the current upper bound for the minimum cut, then it can be contracted. That also means that edges $e$ with $q(e) \geq \hat\lambda$ can be safely contracted, The algorithm is guaranteed to find at least one such edge.

We start this section with optimizations to the sequential algorithm. First we use a recently published \emph{inexact} algorithm to lower the minimum cut upper bound $\hat \lambda$. This enables us to save work and to perform contractions more quickly. 
We then give different implementations of the priority queue $\mathcal{Q}$ and detail the effects of the choice of queue on the algorithm. We show that the algorithm remains correct, even if we limit the priorities in the queue to $\hat\lambda$, meaning that elements in the queue having a key larger than that will not be updated.
This significantly lowers the amount of priority queue operations necessary. Then we adapt the algorithm so that we are able to detect contractible edges in parallel efficiently.
Lastly, we put it everything together and present a full system description.

\subsection{Sequential Optimizations}
\label{seqopt}

\subsubsection{Lowering the Upper Bound $\hat\lambda$.} Note that the upper bound $\hat\lambda$ for the minimum cut is an important parameter for exact contraction based algorithms such as the algorithm \texttt{NOI} of Nagamochi~\etal\cite{nagamochi1994implementing}. The algorithm computes a lower bound for the connectivity of the two incident vertices of each edge and contracts all edges whose incident vertices have a connectivity of at least $\hat\lambda$. Thus it is possible to contract more edges if we manage to lower $\hat\lambda$ beforehand.

A trivial upper bound $\hat\lambda$ for the minimum cut is the minimum vertex degree, as it represents the trivial cut which separates the minimum degree vertex from all other vertices. We run \texttt{VieCut} to lower $\hat\lambda$ in order to allow us to find more edges to contract. Although \texttt{VieCut} is an \emph{inexact algorithm}, in most cases it already finds the minimum cut~\cite{henzinger2018practical} of the graph. As there are by definition no cuts smaller than the minimum cut, the result of \texttt{VieCut} is guaranteed to be at least as large as the minimum cut $\lambda$. As we set $\hat\lambda$ to the result of \texttt{VieCut} when running \texttt{NOI}, we can therefore guarantee~a~correct~result.

A similar idea is employed by the linear time $(2+\epsilon)$-approximation algorithm of Matula~\cite{matula1993linear}, which initializes the algorithm of Nagamochi~\etal\cite{nagamochi1994implementing} with $\hat\lambda = (\frac{1}{2} - \epsilon) \times $min degree.

\subsubsection{Bounded Priority Queues.}
\label{sec:lem}

Whenever we visit a vertex, we update the priority of all of its neighbors in $\mathcal{Q}$ by adding the respective edge weight. Thus, in total we perform $|E|$ priority queue increase-weight operations. In practice, many vertices reach priority values much higher than $\hat\lambda$ and perform many priority increases until they reach their final value. We limit the values in the priority queue by~$\hat\lambda$, \ie we do not update priorities that are already $\hat\lambda$. Lemma~\ref{lem:pq} shows that this does not affect correctness of the algorithm.

Let $\tilde q_G(e)$ be the value $q(e)$ assigned to $e$ in the modified algorithm on graph $G$ and let $\tilde r_G(x)$ be the $r$-value of a node $x$ in the modified algorithm on $G$.

\begin{lemma} \label{lem:pq}
  Limiting the values in the priority queue $\mathcal{Q}$ used in the CAPFOREST routine to a maximum of $\hat\lambda$ does not interfere with the correctness of the algorithm. For every edge $e=(v,w)$ with $\tilde q_G(e) \geq \hat\lambda$, it holds that $\lambda(G,e) \geq \hat\lambda$. Therefore the edge can be contracted.
\end{lemma}

 \begin{proof}
  As we limit the priority queue $\mathcal{Q}$ to a maximum value of $\hat\lambda$, we can not guarantee that we always pop the element with highest value $r(v)$ if there are multiple elements that have values $r(v) \geq \hat\lambda$ in $\mathcal{Q}$. However, we know that the vertex $x$ that is popped from $\mathcal{Q}$ is either maximal or has $r(x) \geq \hat\lambda$.

  We prove Lemma~\ref{lem:pq} by creating a graph $G'=(V,E,c')$ by lowering edge weights (possibly to $0$, effectively removing the edge) while running the algorithm, so that CAPFOREST on $G'$ visits vertices in the same order (assuming equal tie breaking) and assigns the same $q$ values as the modified algorithm~on~$G$.

  We first describe the construction of $G'$. We initialize the weight of all edges in graph $G'$ with the weight of the respective edge in $G$ and run CAPFOREST on $G'$. Whenever we check an edge $e = (x,y)$ and update a value $r_{G'}(y)$ , we check whether we would set $r_{G'}(y) > \hat\lambda$. If this is the case, \ie when $r_{G'}(y) + c(e) > \hat\lambda$, we set $c'(e)$ in $G'$ to $c(e) - (r_{G'}(y) - \hat\lambda)$, which is lower by exactly the value by which $r_{G}(y)$ is larger than $\hat\lambda$, and non-negative. Thus, $r_{G'}(y) = \hat\lambda$. As we scan every edge exactly once in a run of CAPFOREST, the weights of edges already scanned remain constant afterwards. This completes the construction of $G'$

  Note that during the construction of $G'$ edge weights were only decreased and never increased. Thus it holds that $\lambda(G',x,y) \leq \lambda(G,x,y)$ for any pair of nodes $(x,y)$. If we ran the unmodified CAPFOREST algorithm on $G'$ each edge would be assigned a value $q_{G'}(e)$ with $q_{G'}(e) \leq \lambda(G',e)$. Thus for every edge $e$ it holds that $q_{G'}(e) \leq \lambda(G',e) \leq \lambda(G,e)$.

  Below we will show that $\tilde q_G(e) = q_{G'}(e)$ for all edges $e$. It then follows that for all edges $e$ it holds that $\tilde q_G(e) \leq \lambda(G,e)$. This implies that if $\tilde q_G(e) \geq \hat\lambda$ then $\lambda(G,e) \geq \hat\lambda$, which is what we needed to show.

  It remains to show for all edges $e$ that $\tilde q_G(e) = q_{G'}(e)$. To show this claim we will show the following stronger claim. For any $i$ with $1 \leq i \leq m$ after the $(i-1)$th and before the $i$th scan of an edge the modified algorithm on $g$ and the original algorithm on $G'$ with the same tie breaking have visited all nodes and scanned all edges up to now in the same order and for all edges $e$ it holds that $\tilde q_G(e) = q_{G'}(e)$ (we assume that before scanning an edge $e$, $q(e) = 0$) and for all nodes $x$ it holds that $\tilde r_G(x) = r_{G'}(x)$. We show this claim by the induction on $i$.

  For $i=1$ observe that before the first edge scan $\tilde q_G(e) = q_{G'}(e) = 0$ for all edges $e$ and the same node is picked as first node due to identical tie breaking and the fact that $G=G'$ at that point. Now for $i > 1$ assume that the claim holds for $i-1$ and consider the scan of the $(i-1)$th edge. If for the $(i-1)$th edge scan a new node needs to be chosen from the priority queue by one of the algorithms then note that both algorithms will have to choose a node and they pick the same node $y$ as $\tilde r_G(x) = r_{G'}(x)$ for all nodes $x$. Then both algorithms scan the same incident edge of $y$ as in both algorithms the set of unscanned neighbors of $y$ is identical. If neither algorithm has to pick a new node then both have scanned the same edges of the same current node $y$ and due to identical tie breaking will pick the same next edge to scan. Let this edge be $(y,w)$. By induction $\tilde r_G(w) = r_{G'}(w)$ at this time. As $(y,w)$ is unscanned $c'(y,w) = c(y,w)$ which implies that $\tilde r_G(w) + c(y,w) = r_{G'}(w)+c'(y,w)$. If $\tilde r_G(w) + c(y,w) \leq \hat\lambda$ then the modified algorithm on $G$ and the original algorithm on $G'$ will set the $r$ value of $w$ to the same value, namely $\tilde r_G(w) + c(y,w)$. If $\tilde r_G(w) + c(y,w) > \hat\lambda$, then $\tilde r_G(w)$ is set to $\hat\lambda$ and $c'(y,w)$ is set to $c(y,w)-(r_{G'}(w)-\hat\lambda)$, which leads to $r_{G'}(w)$ being set to $\hat\lambda$. Thus $\tilde r_G(w)=r_{G'}(w)$ and by induction $\tilde r_G(x) = r_{G'}(x)$ for all $x$. Additionally the modified algorithm on $G$ sets $\tilde q_G(y,w) = \tilde r_G(w)$ and the original algorithm on $G'$ sets $q_{G'}(y,w)=r_{G'}(w)$. It follows that $\tilde q_G(y,w)=q_{G'}(y,w)$ and, thus, by induction $\tilde q_G(e) = q_{G'}(e)$ for all $e$. This completes the proof of the claim.
\end{proof}

Lemma~\ref{lem:pq} allows us to considerably lower the amount of priority queue operations, as we do not need to update priorities that are bigger than $\hat\lambda$. This optimization has even more benefit in combination with running VieCut to lower the upper bound $\hat\lambda$, as we directly lower the amount of priority queue operations.

\subsubsection{Priority Queue Implementations.}
\label{ssec:pq}

Nagamochi~\etal\cite{nagamochi1994implementing} use an addressable priority queue $\mathcal{Q}$ in their algorithm to find contractible edges. In this section we now address variants for the implementation of the priority queue. As the algorithm often has many elements with maximum priority in practice, the implementation of this priority queue can have major impact on the order of vertex visits and thus also on the edges that will be marked contractible.

\paragraph{Bucket Priority Queue.}
As our algorithm limits the values the priority queue to a maximum of $\hat \lambda$, we observe integer priorities in the range of $[0,\hat\lambda]$. Hence, we can use a bucket queue that is implemented as an array with $\hat\lambda$ buckets. In addition, the data structure keeps the id of the highest non-empty bucket, also known as the \emph{top bucket}, and stores the position of each vertex in the priority queue. Priority updates can be implemented by deleting an element from its bucket and pushing it to the bucket with the updated priority. This allows constant time access for all operations except for deletions of the maximum priority element, which have to check all buckets between the prior top bucket to the new top bucket, possibly up to $\hat\lambda$ checks. We give two possible implementations to implement the buckets so that they can store all elements with a given priority.

The first implementation, \texttt{BStack} uses a dynamic array (\texttt{std::vector}) as the container for all elements in a bucket. When we add a new element to the vector, we push it to the back of the array. \texttt{$\mathcal{Q}$.pop\_max()} returns the last element of the top bucket. Thus our algorithm will always next visit the element whose priority it just increased. The algorithm therefore does not fully explore all vertices in a local region.

The other implementation, \texttt{BQueue} uses a double ended queue (\texttt{std::deque}) as the container instead. A new element is pushed to the back of the queue and \texttt{$\mathcal{Q}$.pop\_max()} returns the \emph{first} element of the top bucket. This results in a variant of our algorithm, which performs closer to a breadth-first search in that it first explores the vertices that have been discovered earlier, \ie are closer to the source vertex in the graph.

\paragraph{Bottom-Up Binary Heap.}
A binary heap~\cite{williams1964heapsort} is a binary tree (implemented as an array, where element $i$ has its children in index $2i$ and $2i+1$) which fulfills the heap property, \ie each element has priority that is not lower than either of its children. Thus the element with highest priority is the root of the tree. The tree can be made addressable by using an array of indices, in which we save the position of each vertex. We use a binary heap using the bottom-up heuristics~\cite{wegener1993bottom}, in which we sift down holes that were created by the deletion of the top priority vertex. Priority changes are implemented by sifting the addressed element up or down in the tree. Operations have a running time of up to $\mathcal{O}(\log n)$ to sift an element up or down to fix the heap property.

In \texttt{$\mathcal{Q}$.pop\_max()}, the \texttt{Heap} priority queue does not favor either old or new elements in the priority queue and therefore this implementation can be seen as a middle ground between the two bucket priority queues.

\begin{algorithm*}[t]
  \begin{algorithmic}[1]
    \INPUT $G = (V,E,c) \leftarrow$ undirected graph $\hat\lambda \leftarrow$ upper bound for minimum cut, $\mathcal{T} \leftarrow$ shared array of vertex visits
    \OUTPUT $\mathcal{U} \leftarrow$ union-find data structure to mark contractible edges
    \State Label all vertices $v \in V$ ``unvisited'', blacklist $\mathcal{B}$ empty
    \State $\forall v \in V: r(v) \leftarrow 0$
    \State $\forall e \in E: q(e) \leftarrow 0$
    \State $\mathcal{Q} \leftarrow$ empty priority queue
    \State Insert random vertex into $\mathcal{Q}$
    \While{$\mathcal{Q}$ not empty}
    \State $x \leftarrow \mathcal{Q}$.pop\_max() \Comment{Choose unvisited vertex with highest priority}
    \State Mark $x$ ``visited''
    \If{$\mathcal{T}(x) = \text{True}$} \Comment{Every vertex is visited only once}
    \State $\mathcal{B}(x) \leftarrow $ True
    \Else
    \State $\mathcal{T}(x) \leftarrow \text{True}$
    \EndIf
    \State $\alpha \leftarrow \alpha + c(x) - 2 r(x)$

    \State $\hat\lambda \leftarrow min(\hat\lambda, \alpha)$
    \For{$e = (x,y) \leftarrow$ edge to vertex $y$ not in $\mathcal{B}$ and not visited}
    \If{$r(y) < \hat\lambda \leq r(y) + c(e)$}
    \State $\mathcal{U}$.union(x,y) \Comment{Mark edge $e$ to contract}
    \EndIf
    \State $r(y) \leftarrow r(y) + c(e)$
    \State $q(e) \leftarrow r(y)$
    \State $\mathcal{Q}(y) \leftarrow min(r(y), \hat\lambda)$

    \EndFor

    \EndWhile

  \end{algorithmic}
  \caption{\label{algo:parnoi} Parallel CAPFOREST}
\end{algorithm*}

\subsection{Parallel CAPFOREST.}
\label{ssec:pmcap}

We modify the algorithm in order to quickly find contractible edges using shared-memory parallelism. The pseudocode can be found in Algorithm~\ref{algo:parnoi}. Pseudocode for the original CAPFOREST algorithm can be found in Algorithm~\ref{algo:cap} in Appendix~\ref{cap}.
The proofs in this section show that the modifications do not violate the correctness of the algorithm.
Detailed proofs for the original CAPFOREST algorithm and the modifications of Nagamochi~\etal for weighted graphs can be~found~in~\cite{nagamochi1994implementing}.

\begin{figure}[t!]
  \centering
  \includegraphics[height=.18\textheight,width=.4\textwidth]{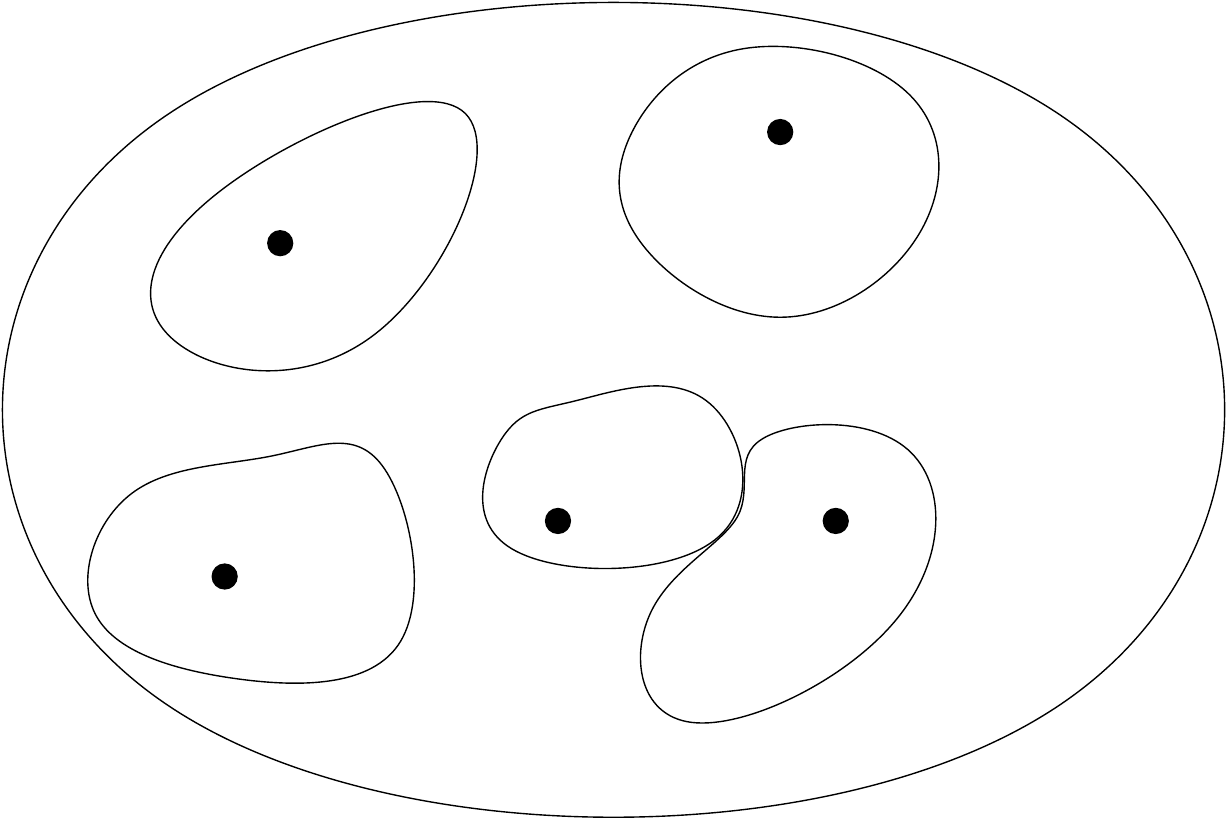}
  \caption{\label{fig:parmc} Example run of Algorithm~\ref{algo:parnoi}. Every process starts at a random vertex and scans region around the start vertex. These regions do not overlap.}  \end{figure}

The idea of the our algorithm is as follows: We aim to find contractible edges using shared-memory parallelism. Every processor selects a random vertex and runs Algorithm~\ref{algo:parnoi}, which is a modified version of CAPFOREST~\cite{nagamochi1992computing,nagamochi1994implementing} where the priority values are limited to $\hat\lambda$, the current upper bound of the size of the minimum cut. We want to find contractible edges without requiring that every process looks at the whole graph. To achieve this, every vertex will only be visited by one process. Compared to limiting the amount of vertices each process visits this has the advantage that we also scan the vertices in sparse regions of the graph which might otherwise not be scanned by any process.

Figure~\ref{fig:parmc} shows an example run of Algorithm~\ref{algo:parnoi} with~$p=5$. Every process randomly chooses a start vertex and performs Algorithm~\ref{algo:parnoi} on it to ``grow a region'' of~scanned~vertices.

As we want to employ shared-memory parallelism to speed up the algorithm, we share an array $\mathcal{T}$ between all processes to denote whether a vertex has already been visited. Every process has a blacklist $\mathcal{B}$ to mark nodes which were already scanned by another process and therefore not explored by this process. For every vertex~$v$ we keep a value $r(v)$, which denotes the total weight of edges connecting $v$ to already scanned vertices. Over the course of a run of the algorithm, every edge~$e~=~(v,w)$ is given a value $q(e)$ (equal to $r(w)$ right after scanning $e$) which is a lower bound for the smallest cut $\lambda(G,v,w)$.  We mark an edge $e$ as contractible (more accurately, we union the incident vertices in the shared concurrent union-find data structure~\cite{anderson1991wait}), if $q(e) \geq \hat\lambda$. Note that this does not modify the graph, it just remembers which nodes to collapse. The actual node collapsing happens in a postprocessing step. Nagamochi and Ibaraki showed \cite{nagamochi1994implementing} that contracting only the edges that fulfill the condition in line $17$ is equivalent.

If a vertex $v$ has already been visited by another process, it will not be visited by any other workers. A process that tries to visit $v$ after it has already been visited locally blacklists $v$ by setting $B(v)$ to true and does not visit the vertex. Subsequently, no more edges incident to $v$ will be marked contractible by this process. This is necessary to ensure correctness of the algorithm. As the set of disconnected edges is different depending on the start vertex, we looked into visiting every vertex by a number of processes up to a given parameter to find more contractible edges. However, this did generally result in higher total running times and thus we only visit every vertex once.

After all processes are finished, every vertex was visited exactly once (or possibly zero times, if the graph is disconnected). On average, every process has visited roughly $\frac{n}{p}$ vertices and all processes finish at the same time. We do not perform any form of locking of the elements of $\mathcal{T}$, as this would come with a running time penalty for every write and the worst that can happen with concurrent writes is that a vertex is visited more often, which does not affect correctness of the algorithm.

However, as we terminate early and no process visits every vertex, we can not guarantee anymore that the algorithm actually finds a contractible edge. However, in practice, this only happens if the graph is already very small ($<50$ vertices in all of our experiments). We can then run the CAPFOREST routine which is guaranteed to find at least one edge to contract.

In line $14$ and $15$ of Algorithm~\ref{algo:parnoi} we compute the value of the cut between the scanned and unscanned vertices and update $\hat\lambda$ if this cut is smaller than it. For more details on this we refer the reader to~\cite{nagamochi1994implementing}.

In practice, many vertices reach values of $r(y)$ that are much higher than $\hat\lambda$ and therefore need to update their priority in $\mathcal{Q}$ often. As previously detailed, we limit the values in the priority queue by $\hat\lambda$ and do not update priorities that are already greater or equal to~$\hat\lambda$. This allows us to considerably lower the amount of priority queue operations per vertex.

\begin{theorem}
  Algorithm~\ref{algo:parnoi} is correct.
\end{theorem}

  Algorithm~\ref{algo:parnoi} is correct. As Algorithm~\ref{algo:parnoi} is a modified variant of CAPFOREST~\cite{nagamochi1992computing,nagamochi1994implementing}, we use the correctness of their algorithm and show that our modifications can not result in incorrect results. In order to show this we need the following lemmas:

  \begin{lemma}
  \begin{enumerate}
    \item Multiple instances of Algorithm~\ref{algo:parnoi} can be run in parallel with all instances sharing a parallel union-find data structure.
    \item Early termination does not affect correctness
    \item For every edge $e = (v,w)$, where neither $v$ nor $w$ are blacklisted, $q(e)$ is a lower bound for the connectivity $\lambda(G, v, w)$, even if the set of blacklisted vertices $\mathcal{B}$ is not empty.
    \item When limiting the priority of a vertex in $\mathcal{Q}$ to $\hat\lambda$, it still holds that the vertices incident to an edge $e = (x,y)$ with $q(e) \geq \hat\lambda$ have connectivity $\lambda(G,x,y) \geq \hat\lambda$.
  \end{enumerate}
\end{lemma}

\begin{proof}
  A run of the CAPFOREST algorithm finds a non-empty set of edges that can be contracted without contracting a cut with value less than $\hat\lambda$~\cite{nagamochi1992computing}. We show that none of our modifications can result in incorrect results:

  \begin{enumerate}
  \item The CAPFOREST routine can be started from an arbitrary vertex and finds a set of edges that can be contracted without affecting the minimum cut $\lambda$. This is true for any vertex $v \in V$. As we do not change the underlying graph but just mark contractible edges, the correctness is obviously upheld when running the algorithm multiple times starting at different vertices. This is also true when running the different iterations in parallel, as long as the underlying graph is not changed.

    Marking the edge $e=(u,v)$ as contractible is equivalent to performing a \texttt{Union} of vertices $u$ and $v$. The \texttt{Union} operation in a union-find data structure is commutative and therefore the order of unions is irrelevant for the final result. Thus performing the iterations successively has the same result as performing them in parallel.

  \item Over the course of the algorithm we set a value $q(e)$ for each edge $e$ and we maintain a value $\hat\lambda$ that never increases. We contract edges that have value~$q(e) \geq \hat\lambda$ at the time when $q(e)$ is set. For every edge, this value is set exactly once. If we terminate the algorithm prior to setting $q(e)$ for all edges, the set of contracted edges is a subset of the set of edges that would be contracted in a full run and all contracted edges $e$ fulfill $q(e) \geq \hat\lambda$ at termination. Thus, no edge contraction contracts a cut that is smaller than $\hat\lambda$.

    \item Let $e=(v,w)$ be an edge and let $\mathcal{B}_e$ be the set of nodes blacklisted at the time when $e$ is scanned. We show that for an edge $e = (v, w)$, $q(e) \leq \lambda(\bar G, v, w)$, where $\bar G = (\bar V, \bar E)$ with vertices $\bar V = V \backslash \mathcal{B}_e$ and edges $\bar E= \{e=(u,v) \in E: u \not\in \mathcal{B}_e $ and $ v \not\in \mathcal{B}_e \}$ is the graph $G$ with all blacklisted vertices and their incident edges removed. As the removal of vertices and edges can not increase edge connectivities $q_{\bar G}(e)\leq\lambda (\bar G, v, w) \leq \lambda(G,v,w)$ and $e$ is a contractible edge.

    Whenever we visit a vertex $b$, we decide whether we blacklist the vertex. If we blacklist the vertex $b$, we immediately leave the vertex and do not change any values $r(v)$ or $q(e)$ for any other vertex or edge. As vertex $b$ is marked as blacklisted, we will not visit the vertex again and the edges incident to $b$ only~affect~$r(b)$.

    As edges incident to any of the vertices in $\mathcal{B}_e$ do not affect~$q(e)$, the value of $q(e)$ in the algorithm with the blacklisted in $G$ is equal to the value of $q(e)$ in $\bar G$, which does not contain the blacklisted vertices in $\mathcal{B}_e$ and their incident edges. On $\bar G$ this is equivalent to a run of CAPFOREST without blacklisted vertices and due to the correctness of CAPFOREST~\cite{nagamochi1994implementing} we know that for every edge $e \in \bar E: q_{\bar G}(e) \leq \lambda(\bar G, v, w) \leq \lambda(G,v,w)$.

    Note that in $\bar G$ we only exclude the vertices that are in $\mathcal{B}_e$. It is possible that a node $y$ that was unvisited when $e$ was scanned might get blacklisted later, however, this does not affect the value of $q(e)$ as the value $q(e)$ is set when an edge is scanned and never modified afterwards.

\item Proof in Lemma~\ref{lem:pq}.

  \end{enumerate}

  We can combine the sub-proofs (3) and (4) by creating the graph $\bar G'$, in which we remove all edges incident to blacklisted vertices and decrease edge weights to make sure no $q(e)$ is strictly larger than $\hat\lambda$. As we only lowered edge weights and removed edges, for every edge between two not blacklisted vertices $e = (u,v)$, $q_G(e) \leq \lambda(\bar G',x,y) \leq \lambda(G,x,y)$ or $q_G(e) > \hat\lambda$ and thus we only contract contractible edges. As none of our modifications can result in the contraction of edges that should not be contracted, Algorithm~\ref{algo:parnoi} is correct.
\end{proof}

\begin{algorithm}[t!]
  \begin{algorithmic}[1]
    \INPUT $G = (V,E,c)$
    \State $\hat\lambda \leftarrow $ \texttt{VieCut}($G$), $G_C \leftarrow G$
    \While{$G_C$ has more than $2$ vertices}
    \State $\hat\lambda \leftarrow$ Parallel CAPFOREST($G_C,\hat\lambda$)
    \If{no edges marked contractible}
        \State  $\hat\lambda \leftarrow$ CAPFOREST($G_C,\hat\lambda$)
        \EndIf
        \State $G_C, \hat\lambda \leftarrow$ Parallel Graph Contract($G_C$)
        \EndWhile

        \State \Return $\hat\lambda$
  \end{algorithmic}
  \caption{\label{algo:parmc} Parallel Minimum Cut}

\end{algorithm}

 \begin{figure*}[b!]
\resizebox{\textwidth}{!}{
  \begin{tikzpicture}
    \begin{axis}[
      name=plot1,
    title={Average Node Degree: $2^5$},
    xlabel={Number of Vertices},
    ylabel={Running Time per Edge [$(ns)$]},
    ymode=log,
    xmode=log,
    ]
\addplot coordinates { (1.04858e+06,82.6098) (2.09715e+06,64.0682) (4.1943e+06,68.9686) (8.38861e+06,71.0581) (1.67772e+07,67.1392) (3.35544e+07,65.6851) };
\addlegendentry{algo=ks\_./misc/CUT/bin/ho};
\addplot coordinates { (1.04858e+06,155.834) (2.09715e+06,76.1902) (4.1943e+06,91.9993) (8.38861e+06,95.6736) (1.67772e+07,78.1644) (3.35544e+07,82.3399) };
\addlegendentry{algo=ks\_./misc/CUT/bin/ni\_nopr};
\addplot coordinates { (1.04858e+06,32.9103) (2.09715e+06,34.1231) (4.1943e+06,31.4984) (8.38861e+06,29.6682) (1.67772e+07,29.234) (3.35544e+07,28.0388) };
\addlegendentry{algo=noibucket,limit};
\addplot coordinates { (1.04858e+06,35.0876) (2.09715e+06,36.2828) (4.1943e+06,32.5431) (8.38861e+06,26.7875) (1.67772e+07,28.5107) (3.35544e+07,25.5876) };
\addlegendentry{algo=noififo,limit};
\addplot coordinates { (1.04858e+06,33.5397) (2.09715e+06,34.7009) (4.1943e+06,33.45) (8.38861e+06,31.4976) (1.67772e+07,31.5349) (3.35544e+07,31.2994) };
\addlegendentry{algo=noiheap};
\addplot coordinates { (1.04858e+06,33.0912) (2.09715e+06,34.8433) (4.1943e+06,32.9472) (8.38861e+06,31.7809) (1.67772e+07,31.785) (3.35544e+07,30.4664) };
\addlegendentry{algo=noiheap,limit};
\addplot coordinates { (1.04858e+06,42.9101) (2.09715e+06,50.6719) (4.1943e+06,49.364) (8.38861e+06,44.1133) (1.67772e+07,47.3919) (3.35544e+07,43.0878) };
\addlegendentry{algo=noiparheap};
\addplot coordinates { (1.04858e+06,49.9442) (2.09715e+06,47.7059) (4.1943e+06,54.1396) (8.38861e+06,50.0579) (1.67772e+07,51.3422) (3.35544e+07,54.1655) };
\addlegendentry{algo=noiparheap,limit};
\legend{}
\end{axis}

\begin{axis}[
  name=plot2,
  at=(plot1.right of south east),
    title={Average Node Degree: $2^6$},
    xlabel={Number of Vertices},
    ymode=log,
    xshift=6mm,
    log basis y={2},
    xmode=log,
    log basis x={2},
    ]
\addplot coordinates { (1.04858e+06,77.1528) (2.09715e+06,71.5663) (4.1943e+06,79.2212) (8.38861e+06,72.4351) (1.67772e+07,74.7277) };
\addlegendentry{algo=ks\_./misc/CUT/bin/ho};
\addplot coordinates { (1.04858e+06,137.518) (2.09715e+06,103.047) (4.1943e+06,144.691) (8.38861e+06,95.9546) (1.67772e+07,95.274) };
\addlegendentry{algo=ks\_./misc/CUT/bin/ni\_nopr};
\addplot coordinates { (1.04858e+06,70.1867) (2.09715e+06,63.2886) (4.1943e+06,51.1212) (8.38861e+06,52.212) (1.67772e+07,50.3563) (3.35544e+07,51.0307) };
\addlegendentry{algo=noibucket,limit};
\addplot coordinates { (1.04858e+06,80.4609) (2.09715e+06,71.4901) (4.1943e+06,59.2622) (8.38861e+06,58.7912) (1.67772e+07,57.0678) (3.35544e+07,57.6975) };
\addlegendentry{algo=noififo,limit};
\addplot coordinates { (1.04858e+06,65.2765) (2.09715e+06,57.2194) (4.1943e+06,47.3124) (8.38861e+06,46.4038) (1.67772e+07,44.6088) (3.35544e+07,45.9924) };
\addlegendentry{algo=noiheap};
\addplot coordinates { (1.04858e+06,65.9292) (2.09715e+06,56.7341) (4.1943e+06,46.8665) (8.38861e+06,46.0783) (1.67772e+07,44.4241) (3.35544e+07,46.0091) };
\addlegendentry{algo=noiheap,limit};
\addplot coordinates { (1.04858e+06,37.4344) (2.09715e+06,37.588) (4.1943e+06,36.4515) (8.38861e+06,41.9957) (1.67772e+07,34.5997) (3.35544e+07,37.0813) };
\addlegendentry{algo=noiparheap};
\addplot coordinates { (1.04858e+06,39.7053) (2.09715e+06,36.8551) (4.1943e+06,38.1337) (8.38861e+06,41.4551) (1.67772e+07,36.5725) (3.35544e+07,36.4937) };
\addlegendentry{algo=noiparheap,limit};

\legend{}
\end{axis}

\begin{axis}[
  name=plot3,
  at=(plot2.right of south east),
  xshift=6mm,
    title={Average Node Degree: $2^7$},
    xlabel={Number of Vertices},
    ymode=log,
    log basis y={2},
    xmode=log,
    log basis x={2},
    legend pos=outer north east,
    legend style={font=\Large},
    label style={font=\Large},
    title style={font=\Large},
    tick label style={font=\Large},
    ]
\addplot coordinates { (1.04858e+06,78.5836) (2.09715e+06,80.4086) (4.1943e+06,72.1249) (8.38861e+06,73.9495) (1.67772e+07,75.3108) };
\addlegendentry{algo=ks\_./misc/CUT/bin/ho};
\addplot coordinates { (1.04858e+06,153.097) (2.09715e+06,166.146) (4.1943e+06,98.9117) (8.38861e+06,105.439) (1.67772e+07,100.357) };
\addlegendentry{algo=ks\_./misc/CUT/bin/ni\_nopr};
\addplot coordinates { (1.04858e+06,82.008) (2.09715e+06,73.9579) (4.1943e+06,57.2105) (8.38861e+06,60.0016) (1.67772e+07,57.1866) (3.35544e+07,59.4987) };
\addlegendentry{algo=noibucket,limit};
\addplot coordinates { (1.04858e+06,94.5967) (2.09715e+06,83.2594) (4.1943e+06,66.7455) (8.38861e+06,67.647) (1.67772e+07,64.7941) (3.35544e+07,67.1024) };
\addlegendentry{algo=noififo,limit};
\addplot coordinates { (1.04858e+06,77.1384) (2.09715e+06,66.8968) (4.1943e+06,53.1356) (8.38861e+06,53.2701) (1.67772e+07,50.2053) (3.35544e+07,53.3701) };
\addlegendentry{algo=noiheap};
\addplot coordinates { (1.04858e+06,77.8587) (2.09715e+06,66.1106) (4.1943e+06,52.3681) (8.38861e+06,53.1689) (1.67772e+07,50.4094) (3.35544e+07,53.4516) };
\addlegendentry{algo=noiheap,limit};
\addplot coordinates { (1.04858e+06,36.9859) (2.09715e+06,37.2847) (4.1943e+06,36.0668) (8.38861e+06,42.734) (1.67772e+07,32.6537) (3.35544e+07,34.5157) };
\addlegendentry{algo=noiparheap};
\addplot coordinates { (1.04858e+06,38.5112) (2.09715e+06,35.6499) (4.1943e+06,37.2433) (8.38861e+06,41.5838) (1.67772e+07,35.3272) (3.35544e+07,34.502) };
\addlegendentry{algo=noiparheap,limit};
\legend{}
\end{axis}

 \begin{axis}[
   name=plot4,
   at=(plot3.right of south east),
   xshift=6mm,
   title={Average Node Degree: $2^8$},
   xlabel={Number of Vertices},
   ymode=log,
   log basis y={2},
   xmode=log,
   log basis x={2},
   legend pos=outer north east,
   legend style={font=\Large},
   label style={font=\Large},
   title style={font=\Large},
   tick label style={font=\Large},
   ]
\addplot coordinates { (1.04858e+06,82.7066) (2.09715e+06,83.8212) (4.1943e+06,87.7874) (8.38861e+06,92.1598) };
\addlegendentry{\texttt{HO-CGKLS}};
\addplot coordinates { (1.04858e+06,168.225) (2.09715e+06,179.57) (4.1943e+06,176.042) (8.38861e+06,176.159) };
\addlegendentry{\texttt{NOI-CGKLS}};
\addplot coordinates { (1.04858e+06,119.102) (2.09715e+06,89.5824) (4.1943e+06,91.4717) (8.38861e+06,131.74) (1.67772e+07,127.045) (3.35544e+07,132.021) };
\addlegendentry{\texttt{NOI$_{\hat\lambda}$-BStack}};
\addplot coordinates { (1.04858e+06,131.953) (2.09715e+06,98.1802) (4.1943e+06,101.698) (8.38861e+06,146.407) (1.67772e+07,141.292) (3.35544e+07,143.677) };
\addlegendentry{\texttt{NOI$_{\hat\lambda}$-BQueue}};
\addplot coordinates { (1.04858e+06,105.726) (2.09715e+06,82.9793) (4.1943e+06,84.048) (8.38861e+06,120.066) (1.67772e+07,117.155) (3.35544e+07,122.471) };
\addlegendentry{\texttt{NOI-HNSS}};
\addplot coordinates { (1.04858e+06,107.612) (2.09715e+06,82.6104) (4.1943e+06,85.2457) (8.38861e+06,120.056) (1.67772e+07,117.162) (3.35544e+07,124.348) };
\addlegendentry{\texttt{NOI$_{\hat\lambda}$-Heap}};
\addplot coordinates { (1.04858e+06,32.3483) (2.09715e+06,28.7585) (4.1943e+06,36.07) (8.38861e+06,31.4417) (1.67772e+07,38.7056) (3.35544e+07,36.768) };
\addlegendentry{\texttt{NOI-HNSS-VieCut}};
\addplot coordinates { (1.04858e+06,32.6933) (2.09715e+06,29.0362) (4.1943e+06,34.6802) (8.38861e+06,31.5548) (1.67772e+07,39.8551) (3.35544e+07,37.4883) };
\addlegendentry{\texttt{NOI$_{\hat\lambda}$-Heap-VieCut}};
\end{axis}
\end{tikzpicture}
}
\caption{Total running time in nanoseconds per edge in RHG graphs. \label{fig:rhg}}
\end{figure*}
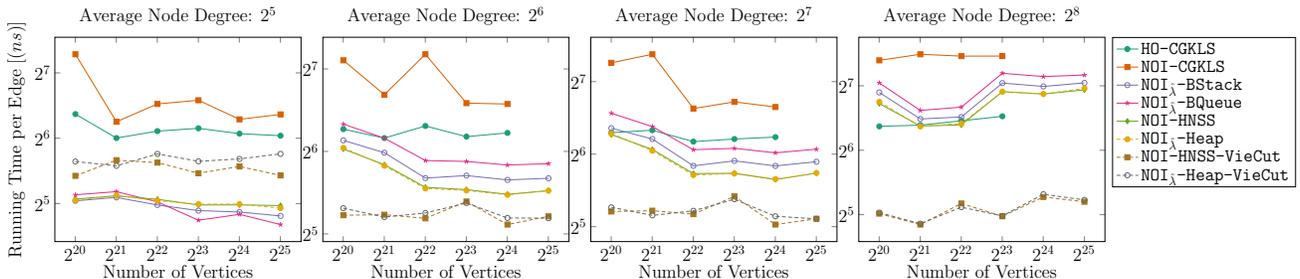%

\subsubsection*{Parallel Graph Contraction.}
\label{contract}

After using Algorithm~\ref{algo:parnoi} to find contractible edges, we use a concurrent hash table~\cite{maier2016concurrent} to generate the contracted graph $G_C = (V_C,E_C)$, in which each block in $\mathcal{U}$ is represented by a single vertex: first we assign each block a vertex ID in the contracted graph in $[0,|V_C|)$. For each edge~$e=(u,v)$, we compute a hash of the block IDs of~$u$ and $v$ to uniquely identify the edge in $E_C$. We use this identifier to compute the weights of all edges between blocks. If there are two blocks that each contain many vertices, there might be many edges between them and if so, the hash table spends considerable time for synchronization. We thus compute the weight of the edge connecting the two heavy blocks locally on each process and sum them up afterwards to reduce synchronization overhead. If the collapsed graph $G_C$ has a minimum degree of less than $\hat\lambda$, we update $\hat\lambda$ to the value~of~this~cut.

\subsection{Putting Things Together.}
\label{putting}

Algorithm~\ref{algo:parmc} shows the overall structure of the algorithm. We first run \texttt{VieCut} to find a good upper bound $\hat\lambda$ for the minimum cut. Afterwards, we run Algorithm~\ref{algo:parnoi} to find contractible edges. In the unlikely case that none were found, we run CAPFOREST~\cite{nagamochi1994implementing} sequentially to find at least one contractible edge. We create a new contracted graph using parallel graph contraction, as shown in Section~\ref{contract}. This process is repeated until the graph has only two vertices left. Whenever we encounter a collapsed vertex with a degree of lower than $\hat\lambda$, we update the upper bound. We return the smallest cut we encounter in this process.

If we also want to output the minimum cut, for each collapsed vertex $v_C$ in $G_C$ we store which vertices of $G$ are included in $v_C$. When we update $\hat\lambda$, we store which vertices are contained in the minimum cut. This allows us to see which vertices are on one side of the cut.

\begin{figure*}[t]
  \centering
  \includegraphics[width=.55\textwidth]{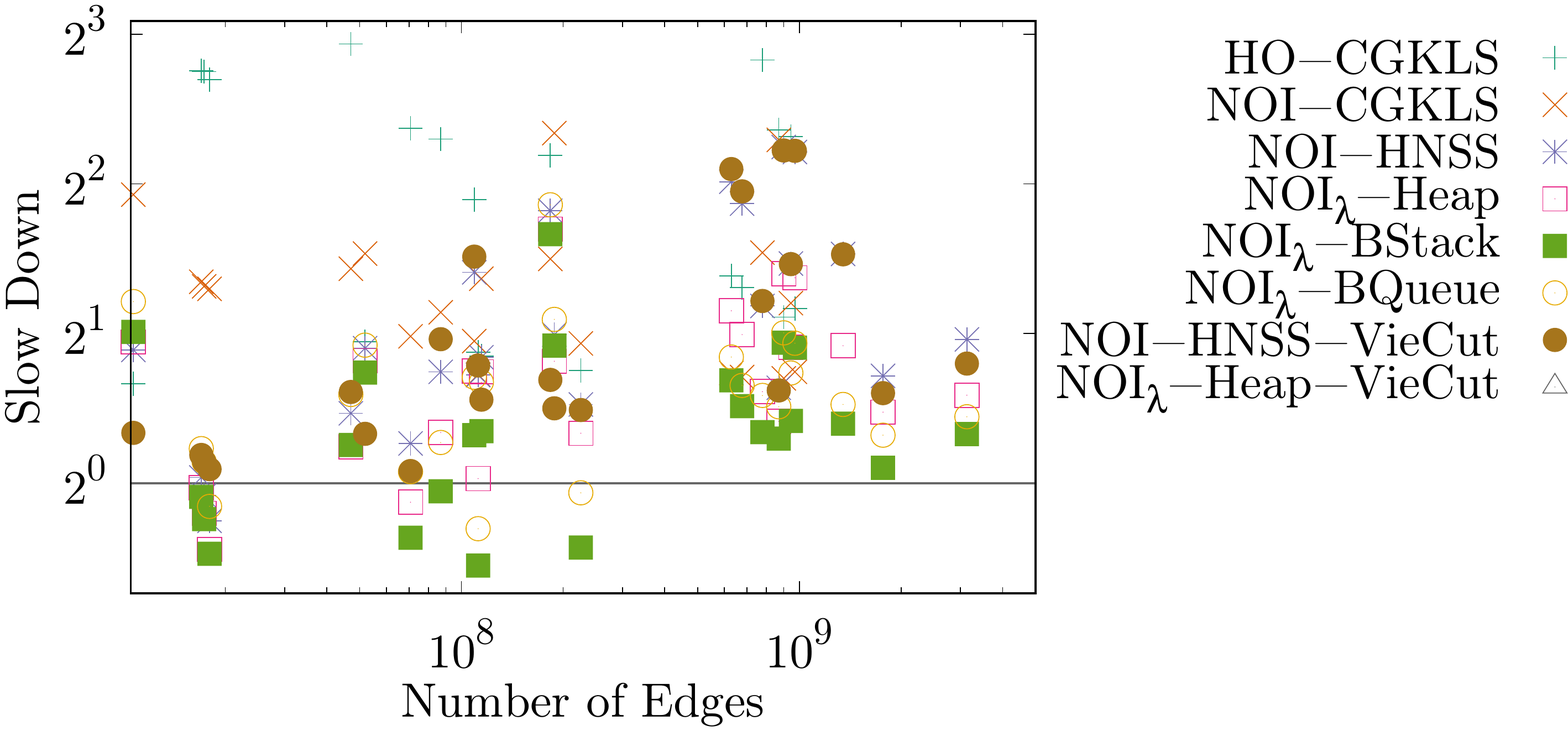}
  \hspace{-3mm}
  \raisebox{3mm}{
    \includegraphics[width=.4\textwidth ]{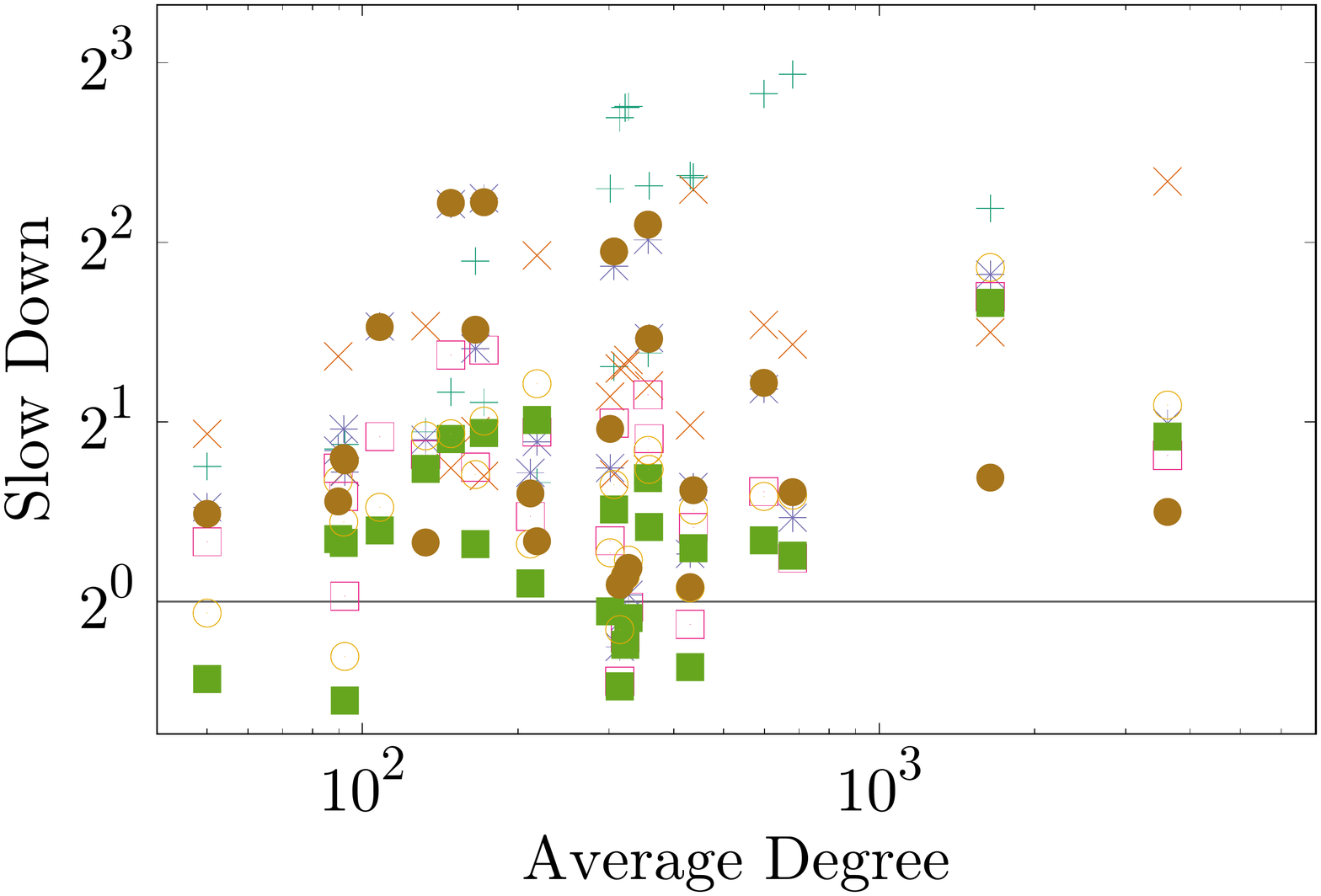}
  }
  \caption{\label{fig:edges} Total running time in real-world graphs, normalized by the running time of \texttt{NOI$_{\hat\lambda}$-Heap-VieCut}.}
\end{figure*}
\begin{figure}[t]
  \centering
  \includegraphics[width=.39\textwidth]{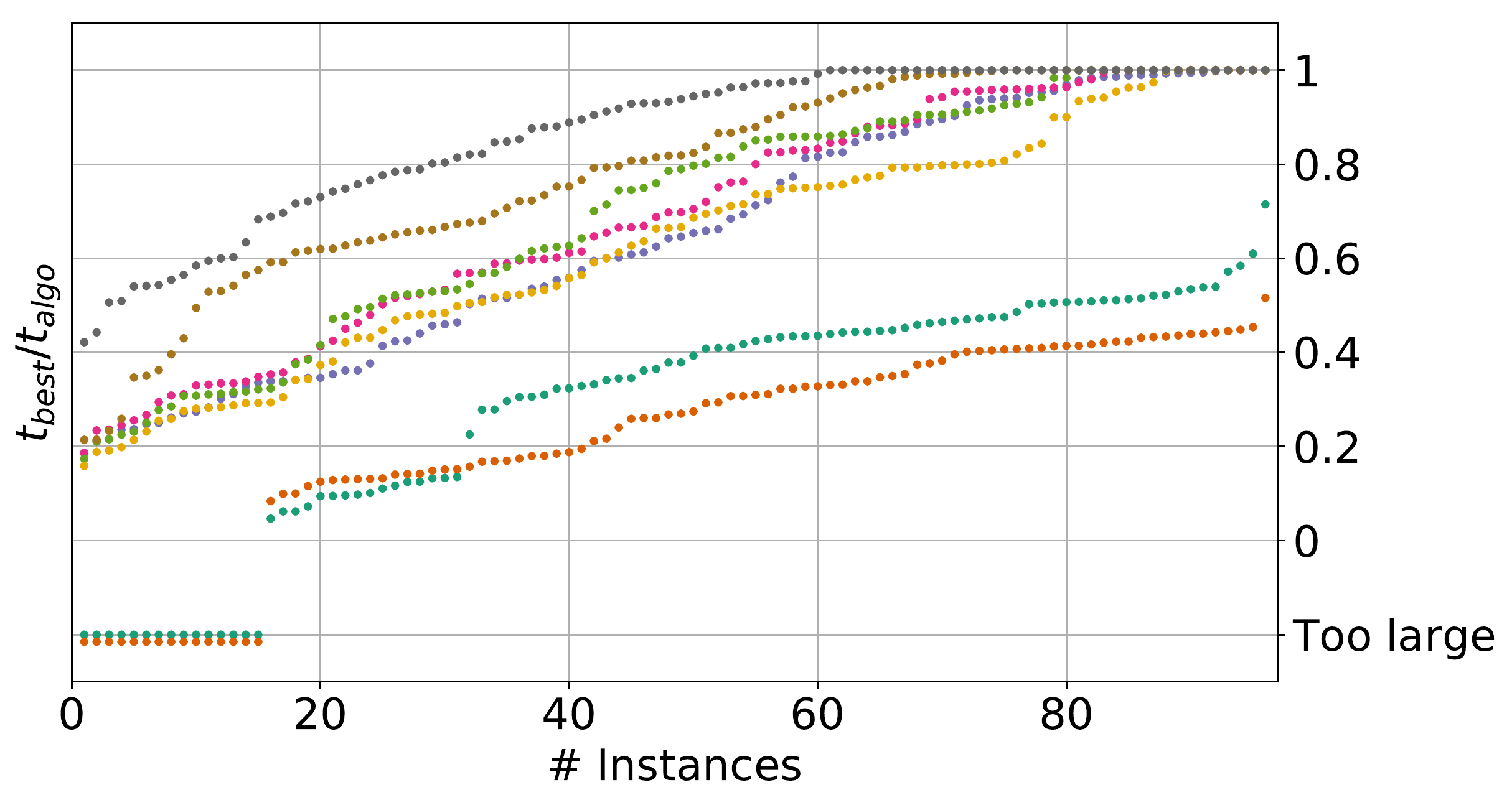}
  \caption{\label{fig:perf} Performance plot for all graphs (legend shared with Figure~\ref{fig:edges} above).}
\end{figure}

\section{Experiments and Results}
\label{impl}

\subsection{Experimental Setup and Methodology.}

We implemented the algorithms using \CC-17 and compiled all
codes using g++-7.1.0 with full optimization (\texttt{-O3}). Our experiments are conducted on a
machine with two Intel Xeon E5-2643 v4 with 3.4GHz with 6 CPU cores each and
1.5 TB RAM in total.
We perform five repetitions per instance and report average running~time.

\emph{Performance plots}
relate the fastest running time to the running time of each other algorithm on a per-instance basis.
For each algorithm, these ratios are sorted in increasing order. The plots show the ratio $t_\text{best}/t_\text{algorithm}$ on the y-axis.
A point close to zero indicates that the running time of the algorithm
was considerably worse than the fastest algorithm on the same instance. A value of one therefore indicates that the corresponding algorithm was one of the fastest algorithms to compute the solution.
Thus an algorithm is considered to outperform another algorithm if its corresponding ratio values are above those of the other algorithm.
In order to include instances that were too big for an algorithm, \ie some implementations are limited to 32bit integers, we set the corresponding ratio below \emph{zero}.

\subsubsection*{Algorithms.} There have been multiple experimental studies that compare exact algorithms for the minimum cut problem~\cite{Chekuri:1997:ESM:314161.314315,henzinger2018practical,junger2000practical}. All of these studies report that the algorithm of Nagamochi~\etal and the algorithm of Hao and Orlin outperform other algorithms, such as the algorithms of Karger and Stein~\cite{karger1996new} or the algorithm of Stoer and Wagner~\cite{stoer1997simple}, often by multiple orders of magnitude. Among others, we compare ourselfs against two available implementations of the sequential algorithm of Nagamochi~\etal\cite{nagamochi1992computing,nagamochi1994implementing}. Henzinger~\etal\cite{henzinger2018practical} give an implementation of the algorithm of Nagamochi~\etal\cite{nagamochi1992computing,nagamochi1994implementing}, written in in \CC~(\texttt{NOI-HNSS}) that uses a binary heap. We use this algorithm with small optimizations in the priority queue as a base of our implementation. Chekuri~\etal\cite{Chekuri:1997:ESM:314161.314315} give an implementation of the flow-based algorithm of Hao and Orlin using all optimizations given in the paper (variant \texttt{ho} in~\cite{Chekuri:1997:ESM:314161.314315}), implemented in C, in our experiments denoted as \texttt{HO-CGKLS}. They also give an implementation of the algorithm of Nagamochi~\etal\cite{nagamochi1992computing,nagamochi1994implementing}, denoted as \texttt{NOI-CGKLS}, which uses a heap as its priority queue data structure (variant \texttt{ni-nopr} in~\cite{Chekuri:1997:ESM:314161.314315}). As their implementations use signed integers as edge ids, we include their algorithms only for graphs that have less than~$2^{31}$~edges. Most of our discussions focus on comparisons to the \texttt{NOI-HNSS} implementation as this outperforms the implementations by~Chekuri~\etal

Gianinazzi~\etal~\cite{gianinazzi2018communication} give a MPI implementation of the algorithm of Karger and Stein~\cite{karger1996new}. We performed preliminary experiments on small graphs which can be solved by \texttt{NOI-HNSS}, \texttt{NOI-CGKLS} and \texttt{HO-CGKLS} in less than $3$ seconds. On these graphs, their implementation using $24$ processes took more than $5$ minutes, which matches other studies~\cite{Chekuri:1997:ESM:314161.314315,junger2000practical,henzinger2018practical} that report bad real-world performance of (other implementations of) the algorithm of Karger and Stein. Gianinazzi~\etal report a running time of $5$ seconds for RMAT graphs with $n=16000$ and an average degree of $4000$, using \emph{$1536$ cores}. As \anoi can find the minimum cut on RMAT graphs~\cite{wsvr} of equal size in less than $2$ seconds using \emph{a single core}, we do not include the implementation in~\cite{gianinazzi2018communication} in our experiments.

As our algorithm solves the minimum cut problem exactly, we do not include the $(2+\epsilon)$-approximation algorithm of Matula~\cite{matula1993linear} and the inexact algorithm \texttt{VieCut} in the experiments.

\begin{figure*}[ht]
\resizebox{\textwidth}{!}{
  \begin{tikzpicture}
    \begin{axis}[
      name=plot1,
    title={\texttt{gsh-2015-host}, $k=10$ ($\lambda=1$)},
    ylabel={Speedup to Sequential},
   xtick={1,2,4,8,12,24},
    xmin=0,
    xmax=25,
    ymax=10,
    ]
\addplot coordinates { (1,1.0) (2,1.49823) (3,2.32523) (4,2.73834) (5,3.52076) (6,3.87512) (7,4.53537) (8,4.81244) (9,5.40642) (10,5.60338) (11,6.13011) (12,6.29108) (24,8.66612) };
\addlegendentry{algo=noiparbucket,limit};
\addplot coordinates { (1,1.0) (2,1.51154) (3,2.35049) (4,2.81597) (5,3.62175) (6,4.00094) (7,4.7307) (8,5.03587) (9,5.69106) (10,5.94895) (11,6.50358) (12,6.75277) (24,9.54242) };
\addlegendentry{algo=noiparfifo,limit};
\addplot coordinates { (1,1.0) (2,1.49041) (3,2.26933) (4,2.67369) (5,3.36341) (6,3.65491) (7,4.26078) (8,4.52316) (9,5.00845) (10,5.18417) (11,5.56771) (12,5.7178) (24,7.93898) };
\addlegendentry{algo=noiparheap,limit};

\addplot [mark=none] coordinates {(1,1) (10,10)};
\legend{}
\end{axis}

\begin{axis}[
  name=plot2,
  at=(plot1.right of south east),
    title={\texttt{uk-2007-05}, $k=10$ ($\lambda=1$)},
    xshift=6mm,
   xtick={1,2,4,8,12,24},
    xmin=0,
    xmax=25,
    ymax=10,
    ]
\addplot coordinates { (1,1.0) (2,1.44457) (3,2.38741) (4,2.78202) (5,3.61606) (6,3.9022) (7,4.72388) (8,5.0009) (9,5.62075) (10,5.83395) (11,6.41581) (12,6.53925) (24,8.68919) };
\addlegendentry{algo=noiparbucket,limit};
\addplot coordinates { (1,1.0) (2,1.45427) (3,2.43538) (4,2.87023) (5,3.77422) (6,4.12578) (7,4.94286) (8,5.22664) (9,5.89796) (10,6.18092) (11,6.78817) (12,6.98936) (24,9.19155) };
\addlegendentry{algo=noiparfifo,limit};
\addplot coordinates { (1,1.0) (2,1.45694) (3,2.29965) (4,2.67507) (5,3.41661) (6,3.71483) (7,4.34879) (8,4.63243) (9,5.15216) (10,5.34193) (11,5.77396) (12,5.94362) (24,8.53983) };
\addlegendentry{algo=noiparheap,limit};

\addplot [mark=none] coordinates {(1,1) (10,10)};
\legend{}
\end{axis}

\begin{axis}[
  name=plot5,
  at=(plot2.right of south east),
    title={\texttt{twitter-2010}, $k=50$ ($\lambda=3$)},
    xshift=6mm,
   xtick={1,2,4,8,12,24},
    xmin=0,
    xmax=25,
    ymax=7,
    ]
\addplot coordinates { (1,1.0) (2,1.33282) (3,2.09342) (4,2.3518) (5,3.0576) (6,3.26268) (7,3.84108) (8,3.98667) (9,4.47033) (10,4.56725) (11,4.98899) (12,5.06745) (24,6.38372) };
\addlegendentry{algo=noiparbucket,limit};
\addplot coordinates { (1,1.0) (2,1.33599) (3,2.09266) (4,2.35684) (5,3.02567) (6,3.24878) (7,3.79514) (8,3.95372) (9,4.4276) (10,4.54551) (11,4.93441) (12,5.03098) (24,6.42067) };
\addlegendentry{algo=noiparfifo,limit};
\addplot coordinates { (1,1.0) (2,1.32377) (3,2.08327) (4,2.32235) (5,2.97119) (6,3.13024) (7,3.60883) (8,3.72148) (9,4.10508) (10,4.15769) (11,4.44833) (12,4.47833) (24,5.75494) };
\addlegendentry{algo=noiparheap,limit};

\addplot [mark=none] coordinates {(1,1) (10,10)};
\legend{}
\end{axis}

\begin{axis}[
  name=plot3,
  at=(plot5.right of south east),
  xshift=6mm,
    title={\texttt{rhg\_25\_8\_1} ($\lambda=118$)},
    legend pos=outer north east,
    legend style={font=\Large},
    label style={font=\Large},
    title style={font=\Large},
    tick label style={font=\Large},
   xtick={1,2,4,8,12,24},
    xmin=0,
    xmax=25,
    ymax=7.5,
    ]
\addplot coordinates { (1,1.0) (2,1.39976) (3,2.24775) (4,2.73438) (5,3.52622) (6,3.94708) (7,4.6326) (8,4.90194) (9,5.50095) (10,5.90176) (11,6.27589) (12,6.51512) (24,6.49556) };
\addlegendentry{algo=noiparbucket,limit};
\addplot coordinates { (1,1.0) (2,1.3616) (3,2.2223) (4,2.65557) (5,3.49934) (6,3.79186) (7,4.59682) (8,4.83102) (9,5.41411) (10,5.68156) (11,6.22205) (12,6.42516) (24,6.75327) };
\addlegendentry{algo=noiparfifo,limit};
\addplot coordinates { (1,1.0) (2,1.40166) (3,2.26015) (4,2.72589) (5,3.61387) (6,3.95402) (7,4.74889) (8,5.18733) (9,5.71225) (10,6.04196) (11,6.41044) (12,6.64743) (24,6.29196) };
\addlegendentry{algo=noiparheap,limit};

\addplot [mark=none] coordinates {(1,1) (10,10)};
\legend{}
\end{axis}

 \begin{axis}[
   name=plot4,
   at=(plot3.right of south east),
   xshift=6mm,
    title={\texttt{rhg\_25\_8\_2} ($\lambda=73$)},
   legend pos=outer north east,
   legend style={font=\Large},
   label style={font=\Large},
   title style={font=\Large},
   tick label style={font=\Large},
   xtick={1,2,4,8,12,24},
    xmin=0,
    xmax=25,
    ymax=7.5,
   ]

\addplot coordinates { (1,1.0) (2,1.37499) (3,2.45014) (4,2.80537) (5,3.70023) (6,4.01833) (7,4.68553) (8,5.10443) (9,5.55565) (10,5.97449) (11,6.38261) (12,6.41539) (24,5.69074) };
\addlegendentry{\texttt{ParCut$_{\hat\lambda}$-BStack}};
\addplot coordinates { (1,1.0) (2,1.34478) (3,2.44026) (4,2.70384) (5,3.52056) (6,3.90222) (7,4.82463) (8,5.20149) (9,5.72159) (10,6.01846) (11,6.59406) (12,6.33723) (24,6.04714) };
\addlegendentry{\texttt{ParCut$_{\hat\lambda}$-BQueue}};
\addplot coordinates { (1,1.0) (2,1.39195) (3,2.40161) (4,2.94518) (5,3.77167) (6,4.09361) (7,4.96292) (8,5.57365) (9,5.73942) (10,6.29447) (11,6.66352) (12,6.69861) (24,5.71081) };
\addlegendentry{\texttt{ParCut$_{\hat\lambda}$-Heap}};

\addplot [mark=none] coordinates {(1,1) (10,10)};
\end{axis}
\end{tikzpicture}
}
\label{fig:scale}
\end{figure*}

\begin{figure*}
\resizebox{\textwidth}{!}{
  \begin{tikzpicture}
    \begin{axis}[
      name=plot1,
    xlabel={Number of Processes},
    ylabel={Speedup to Sequential},
    xtick={1,2,4,8,12,24},
    xmin=0,
    xmax=25,
    ]
\addplot coordinates { (1,0.779197) (2,1.16741) (3,1.81181) (4,2.13371) (5,2.74337) (6,3.01948) (7,3.53394) (8,3.74983) (9,4.21266) (10,4.36614) (11,4.77656) (12,4.90198) (24,6.75261) };
\addlegendentry{algo=noiparbucket,limit};
\addplot coordinates { (1,0.765796) (2,1.15753) (3,1.8) (4,2.15646) (5,2.77352) (6,3.06391) (7,3.62275) (8,3.85645) (9,4.35819) (10,4.55569) (11,4.98042) (12,5.17124) (24,7.30755) };
\addlegendentry{algo=noiparfifo,limit};
\addplot coordinates { (1,0.736603) (2,1.09784) (3,1.6716) (4,1.96945) (5,2.4775) (6,2.69222) (7,3.13851) (8,3.33177) (9,3.68924) (10,3.81868) (11,4.1012) (12,4.21175) (24,5.84788) };
\addlegendentry{algo=noiparheap,limit};

\addplot[black,mark=none] coordinates { (0,1.0) (25,1.0) };
\addlegendentry{\texttt{ParCut$_{\infty}$-Heap}};
\addplot[mark=none,red] coordinates { (0,2.0) };
\addlegendentry{\texttt{NOI$_{\hat\lambda}$-Heap}};
  \addplot[mark=none,brown] coordinates { (0,2.20) (25,2.20) };
  \addlegendentry{\texttt{NOI$_{\hat\lambda}$-BStack}};

\legend{}
\end{axis}

\begin{axis}[
  name=plot2,
  at=(plot1.right of south east),
    xlabel={Number of Processes},
    xshift=6mm,
    xtick={1,2,4,8,12,24},
    xmin=0,
    xmax=25,
    ]

\addplot coordinates { (1,0.391996) (2,0.566266) (3,0.935855) (4,1.09054) (5,1.41748) (6,1.52965) (7,1.85175) (8,1.96033) (9,2.20331) (10,2.28689) (11,2.51497) (12,2.56336) (24,3.40613) };
\addlegendentry{algo=noiparbucket,limit};
\addplot coordinates { (1,0.390661) (2,0.568126) (3,0.95141) (4,1.12129) (5,1.47444) (6,1.61178) (7,1.93098) (8,2.04184) (9,2.3041) (10,2.41464) (11,2.65187) (12,2.73047) (24,3.59078) };
\addlegendentry{algo=noiparfifo,limit};
\addplot coordinates { (1,0.36251) (2,0.528154) (3,0.833646) (4,0.969739) (5,1.23855) (6,1.34666) (7,1.57648) (8,1.6793) (9,1.86771) (10,1.9365) (11,2.09312) (12,2.15462) (24,3.09577) };
\addlegendentry{algo=noiparheap,limit};

\addplot[black,mark=none] coordinates { (0,1.0) (25,1.0) };
\addlegendentry{\texttt{ParCut$_{\infty}$-Heap}};
\addplot[mark=none,red] coordinates { (0,2.0) };
\addlegendentry{\texttt{NOI$_{\hat\lambda}$-Heap}};
  \addplot[mark=none,brown] coordinates { (0,1.56) (25,1.56) };
\addlegendentry{\texttt{NOI$_{\hat\lambda}$-BStack}};
\legend{}
\end{axis}

\begin{axis}[
  name=plot5,
  at=(plot2.right of south east),
    xlabel={Number of Processes},
    xshift=6mm,
   xtick={1,2,4,8,12,24},
    xmin=0,
    xmax=25,
    ]
\addplot coordinates { (1,1.61061) (2,2.14666) (3,3.37169) (4,3.78784) (5,4.92461) (6,5.25491) (7,6.18649) (8,6.42098) (9,7.19997) (10,7.35606) (11,8.03532) (12,8.1617) (24,10.2817) };
\addlegendentry{algo=noiparbucket,limit};
\addplot coordinates { (1,1.5922) (2,2.12716) (3,3.33194) (4,3.75256) (5,4.81748) (6,5.17271) (7,6.04263) (8,6.29512) (9,7.04963) (10,7.23738) (11,7.85657) (12,8.01034) (24,10.223) };
\addlegendentry{algo=noiparfifo,limit};
\addplot coordinates { (1,1.4869) (2,1.96831) (3,3.09762) (4,3.4531) (5,4.41785) (6,4.65434) (7,5.36596) (8,5.53345) (9,6.10383) (10,6.18206) (11,6.61421) (12,6.65882) (24,8.557) };
\addlegendentry{algo=noiparheap,limit};

\addplot[black,mark=none] coordinates { (0,1.0) (25,1.0) };
\addlegendentry{\texttt{ParCut$_{\infty}$-Heap}};
\addplot[mark=none,red] coordinates { (0,2.0) };
\addlegendentry{\texttt{NOI$_{\hat\lambda}$-Heap}};
  \addplot[mark=none,brown] coordinates { (0,2.55) (25,2.55) };
\addlegendentry{\texttt{NOI$_{\hat\lambda}$-BStack}};

\legend{}
\end{axis}

\begin{axis}[
  name=plot3,
  at=(plot5.right of south east),
  xshift=6mm,
    xlabel={Number of Processes},
    legend pos=outer north east,
    legend style={font=\Large},
    label style={font=\Large},
    title style={font=\Large},
    tick label style={font=\Large},
   xtick={1,2,4,8,12,24},
    xmin=0,
    xmax=25,
    ]

\addplot coordinates { (1,0.903528) (2,1.26473) (3,2.0309) (4,2.47059) (5,3.18604) (6,3.5663) (7,4.18569) (8,4.42903) (9,4.97026) (10,5.3324) (11,5.67044) (12,5.88659) (24,5.86892) };
\addlegendentry{algo=noiparbucket,limit};
\addplot coordinates { (1,0.852676) (2,1.161) (3,1.89491) (4,2.26434) (5,2.9838) (6,3.23323) (7,3.9196) (8,4.1193) (9,4.61648) (10,4.84453) (11,5.3054) (12,5.47858) (24,5.75835) };
\addlegendentry{algo=noiparfifo,limit};
\addplot coordinates { (1,0.975365) (2,1.36713) (3,2.20447) (4,2.65874) (5,3.52484) (6,3.85661) (7,4.63191) (8,5.05954) (9,5.57153) (10,5.89311) (11,6.25252) (12,6.48367) (24,6.13695) };
\addlegendentry{algo=noiparheap,limit};

\addplot[black,mark=none] coordinates { (0,1.0) (25,1.0) };
\addlegendentry{\texttt{ParCut$_{\infty}$-Heap}};
\addplot[mark=none,red] coordinates { (0,0.99) (25,0.99) };
\addlegendentry{\texttt{NOI$_{\hat\lambda}$-Heap}};
  \addplot[mark=none,brown] coordinates { (0,0) };
\addlegendentry{\texttt{NOI$_{\hat\lambda}$-BStack}};

\legend{}
\end{axis}

 \begin{axis}[
   name=plot4,
   at=(plot3.right of south east),
   xshift=6mm,
    xlabel={Number of Processes},
   legend pos=outer north east,
   legend style={font=\Large},
   label style={font=\Large},
   title style={font=\Large},
   tick label style={font=\Large},
   xtick={1,2,4,8,12,24},
    xmin=0,
    xmax=25,
    ymin=0.001,
   ]

   \addplot coordinates { (1,1.87417) (2,2.57696) (3,4.59197) (4,5.25772) (5,6.93484) (6,7.53101) (7,8.78145) (8,9.56655) (9,10.4122) (10,11.1972) (11,11.9621) (12,12.0235) (24,10.6654) };
\addlegendentry{\texttt{ParCut$_{\hat\lambda}$-BStack}};
\addplot coordinates { (1,1.75653) (2,2.36214) (3,4.28638) (4,4.74937) (5,6.18395) (6,6.85435) (7,8.47459) (8,9.13656) (9,10.0501) (10,10.5716) (11,11.5826) (12,11.1315) (24,10.622) };
\addlegendentry{\texttt{ParCut$_{\hat\lambda}$-BQueue}};
\addplot coordinates { (1,1.93646) (2,2.69546) (3,4.65062) (4,5.70323) (5,7.30368) (6,7.92712) (7,9.6105) (8,10.7931) (9,11.1142) (10,12.189) (11,12.9036) (12,12.9716) (24,11.0588) };
\addlegendentry{\texttt{ParCut$_{\hat\lambda}$-Heap}};

\addplot[black,mark=none] coordinates { (0,1.0) (25,1.0) };
\addlegendentry{\texttt{NOI-HNSS}};
\addplot[mark=none,red] coordinates { (0,0.99) (25,0.99) };
\addlegendentry{\texttt{NOI$_{\hat\lambda}$-Heap}};
  \addplot[mark=none,brown] coordinates { (0,0) };
\addlegendentry{\texttt{NOI$_{\hat\lambda}$-BStack}};
\end{axis}
\end{tikzpicture}
}
\caption{\label{fig:scale2} Scaling plots for four large graphs. Top: Scalability. Bottom: Speedup compared to \texttt{NOI-HNSS} and fastest sequential algorithm (first $3$ graphs: \texttt{NOI$_{\hat\lambda}$-BStack}¸ last $2$ graphs: \texttt{NOI$_{\hat\lambda}$-Heap}). }
\end{figure*}
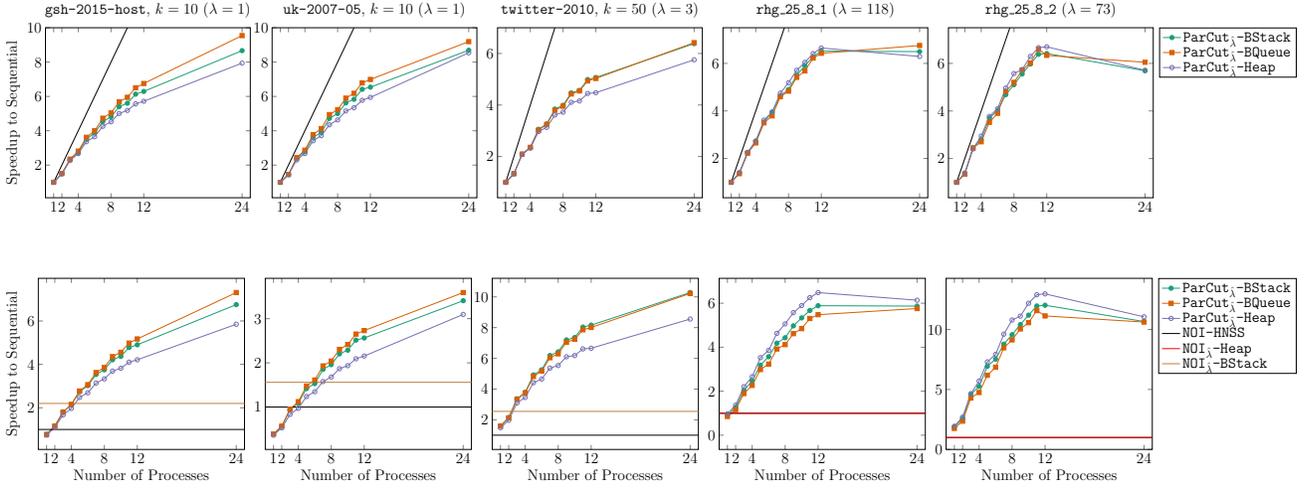

\subsubsection*{Instances.}

We use a set of graph instances from the experimental study of Henzinger~\etal\cite{henzinger2018practical}. The set of instances contains $k$-cores~\cite{batagelj2003m} of large undirected real-world graphs taken from the 10th DIMACS Implementation Challenge~\cite{benchmarksfornetworksanalysis} as well as the Laboratory for Web Algorithmics~\cite{BRSLLP,BoVWFI}. Additionally it contains large random hyperbolic graphs~\cite{krioukov2010hyperbolic, von2015generating} with $n=2^{20} - 2^{25}$ and $m=2^{24} - 2^{32}$. A detailed description of the graph instances is given~in~Appendix~\ref{app:instances}. These graphs are unweighted, however contracted graphs that are created in the course of the algorithm have edge weights.

\subsection{Sequential Experiments.}
\label{exp:seq}

We limit the values in the priority queue $\mathcal{Q}$ to $\hat\lambda$, in order to significantly lower the amount of priority queue operations needed to run the contraction routine. In this experiment, we want to examine the effects of different priority queue implementations and limiting priority queue values have on sequential minimum cut computations. We also include variants which run \texttt{VieCut} first to lower $\hat\lambda$.

We start with sequential experiments using the implementation of \texttt{NOI-HNSS}. We use two variants: \texttt{NOI}$_{\hat\lambda}$ limits values in the priority queue to $\hat\lambda$ while \texttt{NOI-HNSS} allows arbitrarily large values in $\mathcal{Q}$. For \texttt{NOI}$_{\hat\lambda}$, we test the three priority queue implementations, \texttt{BQueue}, \texttt{Heap} and \texttt{BStack}. As the priority queue for \texttt{NOI-HNSS} has priorities of up to the maximum degree of the graph and the contracted graphs can have very large degrees, the bucket priority queues are not suitable for \texttt{NOI-HNSS}. Therefore we only use the implementation of \texttt{NOI-HNSS}~\cite{henzinger2018practical}.
  The variants \texttt{NOI-HNSS-VieCut} and \texttt{NOI$_{\hat\lambda}$-Heap-VieCut} first run the shared-memory parallel algorithm \texttt{VieCut} using all $24$ threads to lower $\hat\lambda$ before running the respective sequential algorithm. We report the total running time, \eg the sum of \texttt{VieCut} and \texttt{NOI}.

\paragraph{Priority Queue Implementations.}

Figure~\ref{fig:rhg} shows the results for RHG graphs and Figure~\ref{fig:edges} shows the results for real-world graphs, normalized by the running time of \texttt{NOI$_{\hat\lambda}$-Heap-VieCut}. Figure~\ref{fig:perf} gives performance plots for all graphs. We can see that in nearly all sequential runs, \texttt{NOI$_{\hat\lambda}$-BStack} is $5-10\%$ faster than  \texttt{NOI$_{\hat\lambda}$-BQueue}. This can be explained as this priority queue uses \texttt{std::vector} instead of \texttt{std::deque} as its underlying data structure and thus has lower access times to add and remove elements. As all vertices are visited by the only thread, the scan order does not greatly influence how many edges~are~contracted.

In the RHG graphs, nearly no vertices in \texttt{NOI-HNSS} reach priorities in $\mathcal{Q}$ that are much larger than $\hat\lambda$. Usually, less than $5\%$ of edges do not incur an update in $\mathcal{Q}$. Thus, \texttt{NOI-HNSS} and \texttt{NOI$_{\hat\lambda}$-Heap} have practically the same running time.  \texttt{NOI$_{\hat\lambda}$-BStack} is usually $5\%$ slower.

As the real-world graphs are social network and web graphs, they contain vertices with very high degrees. In these vertices, \texttt{NOI-HNSS} often reaches priority values of much higher than $\hat\lambda$ and \texttt{NOI$_{\hat\lambda}$} can actually save priority queue operations. Thus, \texttt{NOI$_{\hat\lambda}$-Heap} is up to $1.83$ times faster than \texttt{NOI-HNSS} with an average (geometric) speedup factor of $1.35$. Also, in contrast to the RHG graphs, \texttt{NOI$_{\hat\lambda}$-BStack} is faster than \texttt{NOI-HNSS} on real-world graphs. Due to the low diameter of web and social graphs, the number of vertices in $\mathcal{Q}$ is very large. This favors the \texttt{BStack} priority queue, as it has constant~access~times. The average geometric speedup of \texttt{NOI$_{\hat\lambda}$-BStack} compared to \texttt{NOI$_{\hat\lambda}$-Heap} is $1.22$.

\paragraph{Reduction of $\hat\lambda$ by \texttt{VieCut}.}

Now we reduce $\hat\lambda$ by \texttt{VieCut} before running \texttt{NOI}. While the other algorithms are slower for denser RHG graphs, \texttt{NOI-HNSS-VieCut} and \texttt{NOI$_{\hat\lambda}$-Heap-VieCut} are faster in these graphs with higher density. This happens as the variants without \texttt{VieCut} find less contractible edges and therefore need more rounds of CAPFOREST. The highest speedup compared to \texttt{NOI$_{\hat\lambda}$-Heap} is reached in RHG graphs with $n=2^{23}$ and an average density of $2^8$, where  \texttt{NOI$_{\hat\lambda}$-Heap-VieCut} has a speedup of factor $4$.

\texttt{NOI$_{\hat\lambda}$-Heap-VieCut} is fastest on most real-world graphs, however when the minimum degree is very close to the minimum cut $\lambda$, running \texttt{VieCut} can not significantly lower $\hat\lambda$. Thus, the extra work to run \texttt{VieCut} takes longer than the time saved by lowering the upper bound $\hat\lambda$. The average geometric speedup factor of \texttt{NOI$_{\hat\lambda}$-Heap-VieCut} on all graphs compared to the variant without \texttt{VieCut} is $1.34$.

In the performance plots in Figure~\ref{fig:perf} we can see that \texttt{NOI$_{\hat\lambda}$-Heap-VieCut} is fastest or close to the fastest algorithm in all but the very sparse graphs, in which the algorithm of Nagamochi~\etal\cite{nagamochi1994implementing} is already very fast~\cite{henzinger2018practical} and therefore using \texttt{VieCut} cannot sufficiently lower $\hat\lambda$ and thus the running time of the algorithm. \texttt{NOI-CGKLS} and \texttt{HO-CGKLS} are outperformed on all graphs.

\subsection{Shared-memory parallelism.}

We run experiments on $5$ of the largest graphs in the data sets using up to $24$ threads on $12$ cores. First, we compare the performance of Algorithm~\ref{algo:parmc} using different priority queues: \texttt{ParCut$_{\hat\lambda}$-Heap}, \texttt{ParCut$_{\hat\lambda}$-BStack} and \texttt{ParCut$_{\hat\lambda}$-BQueue} all limit the priorities to $\hat\lambda$, the result of \texttt{VieCut}.

Figure~\ref{fig:scale2} shows the results of these scaling experiments. The top row shows how well the algorithms scale with increased amounts of processors. The lower row shows the speedup compared to the fastest sequential algorithm of Section~\ref{exp:seq}. On all graphs, \texttt{ParCut$_{\hat\lambda}$-BQueue} has the highest speedup when using $24$ threads. On real-world graphs, \texttt{ParCut$_{\hat\lambda}$-BQueue} also has the lowest total running time. In the large RHG graphs, in which the priority queue is usually only filled with up to $1000$ elements, the worse constants of the double-ended queue cause the variant to be slightly slower than \texttt{ParCut$_{\hat\lambda}$-Heap} also even when running with $24$ threads. In the two large real-world graphs that have a minimum degree of $10$, the sequential algorithm \texttt{NOI$_{\hat\lambda}$-BStack} contracts most edges in a single run of CAPFOREST - due to the low minimum degree, the priority queue operations per vertex are also very low. Thus, \texttt{ParCut$_{\hat\lambda}$} using only a single thread has a significantly higher running time, as it runs \texttt{VieCut} first and performs graph contraction using a concurrent hash table, as described in Section~\ref{contract}, which is slower than sequential graph contraction when using just one thread. In graphs with higher minimum degree, \texttt{NOI} needs to perform multiple runs of CAPFOREST. By lowering $\hat\lambda$ using \texttt{VieCut} we can contract significantly more edges and achieve a speedup factor of up to $12.9$ compared to the fastest sequential algorithm \texttt{NOI$_{\hat\lambda}$-Heap}.
On \texttt{twitter-2010}, $k=50$, \texttt{ParCut$_{\hat\lambda}$-BQueue} has a speedup of $10.3$ to \texttt{NOI-HNSS}, $16.8$ to \texttt{NOI-CGKLS} and a speedup of $25.5$ to \texttt{HO-CGKLS}. The other graphs have more than $2^{31}$ edges and are thus too large for \texttt{NOI-CGKLS} and \texttt{HO-CGKLS}.

\section{Conclusion}
\label{conclusion}

We presented a shared-memory parallel exact algorithm for the minimum cut problem. Our algorithm is based on the algorithms of Nagamochi~\etal~\cite{nagamochi1992computing,nagamochi1994implementing} and of Henzinger~\etal~\cite{henzinger2018practical}. We use different data structures and optimizations to decrease the running time of the algorithm of Nagamochi~\etal by a factor of up to $2.5$. Using additional shared-memory parallelism we further increase the speedup factor to up to $12.9$. Future work includes checking whether our sequential optimizations and parallel implementation can be applied to the $(2+\epsilon)$-approximation algorithm of Matula~\cite{matula1993linear}.

\bibliographystyle{abbrv}
\bibliography{quellen}

\begin{appendix}

\section{Instances and Capforest Pseudocode}
\label{app:instances}

\subsection{Random Hyperbolic Graphs (RHG) \cite{krioukov2010hyperbolic}.}

\label{rhggraphs}

Random hyperbolic graphs replicate many features of
real-world networks~\cite{chakrabarti2006graph}: the degree distribution follows
a power law, they often exhibit a community structure and have a small
diameter. In denser hyperbolic graphs, the minimum cut is often equal to the
minimum degree, which results in a trivial minimum cut. In order to prevent
trivial minimum cuts, we use a power law exponent of $5$. We use the generator
of von Looz \etal\cite{von2015generating}, which is a part of NetworKit~\cite{staudt2014networkit}, to generate unweighted random hyperbolic
graphs with $2^{20}$ to $2^{25}$ vertices and an average vertex degree of $2^5$
to $2^8$. These graphs generally have very few small cuts and the minimum
cut has two partitions with similar sizes.

\begin{table*}[t]
  \setlength\intextsep{0pt}
  \centering
  \begin{tabular}{l|r|r||r|r|r|r|r}
    graph & $n$ & $m$ & $k$ & $n$ & $m$ & $\lambda$ & $\delta$ \\\hline\hline
    \texttt{hollywood-2011} & 2.2M & 114M & 20 & 1.3M & 109M & 1 & 20 \\
     \cite{BRSLLP,BoVWFI}    &&& 60 & 576K & 87M & 6 & 60\\
         &&& 100 & 328K & 71M & 77 & 100\\
         &&& 200 & 139K & 47M & 27 & 200\\\hline
    \texttt{com-orkut} & 3.1M & 117M & 16 & 2.4M & 112M & 14 & 16 \\
    \cite{BRSLLP,BoVWFI}    &&& 95 & 114K & 18M & 89 & 95\\
         &&& 98 & 107K & 17M & 76 & 98\\
         &&& 100 & 103K & 17M & 70 & 100\\\hline
    \texttt{uk-2002} & 18M & 262M & 10 & 9M & 226M & 1 & 10\\
    \cite{benchmarksfornetworksanalysis,BRSLLP,BoVWFI}     &&& 30 & 2.5M & 115M & 1 & 30\\
         &&& 50 & 783K & 51M & 1 & 50\\
         &&& 100 & 98K & 11M & 1 & 100\\\hline
    \texttt{twitter-2010} & 42M & 1.2B & 25 & 13M & 958M & 1 & 25 \\
    \cite{BRSLLP,BoVWFI}     &&& 30 & 10M & 884M & 1 & 30\\
         &&& 50 & 4.3M & 672M & 3 & 50\\
         &&& 60 & 3.5M & 625M & 3 & 60\\\hline
    \texttt{gsh-2015-host} & 69M & 1.8B & 10 & 25M & 1.3B & 1 & 10 \\
    \cite{BRSLLP,BoVWFI}     &&& 50 & 5.3M & 944M & 1 & 50\\
         &&& 100 & 2.6M & 778M & 1 & 100\\
         &&& 1000 & 104K & 188M & 1 & 1000\\\hline
    \texttt{uk-2007-05} & 106M & 3.3B & 10 & 68M & 3.1B & 1 & 10\\
    \cite{benchmarksfornetworksanalysis,BRSLLP,BoVWFI}     &&& 50 & 16M & 1.7B & 1 & 50 \\
         &&& 100 & 3.9M & 862M & 1 & 100\\
         &&& 1000 & 222K & 183M & 1 & 1000\\\hline
  \end{tabular}
  \caption{Statistics of real-world web graphs used in experiments. Original graph size and $k$-cores used in experiments with their respective minimum cuts\label{table:realworld}}
\end{table*}

\subsection{Real-world Graphs.}
\label{rwgraphs}

We use large real-world web graphs and social
networks from~\cite{benchmarksfornetworksanalysis,BRSLLP,BoVWFI}, detailed in Table~\ref{table:realworld}. The minimum cut
problem on these web and social graphs can be seen as a network reliability
problem. As these graphs are generally disconnected and contain vertices with
very low degree, we use a $k$-core
decomposition~\cite{seidman1983network,batagelj2003m} to generate
versions of the graphs with a minimum degree of $k$. The $k$-core of a graph $G =
(V, E)$ is the maximum subgraph $G' = (V',E')$ with $V' \subseteq V$ and $E'
\subseteq E$, which fulfills the condition that every vertex in $G'$ has a
degree of at least $k$. We perform our experiments on the largest connected
component of $G'$. For every real-world graph we use, we compute a set of $4$
different $k$-cores, in which the minimum cut is not equal to the minimum degree.

We generate a diverse set of graphs with different sizes. On the large
graphs \texttt{gsh-2015-host} and \texttt{uk-2007-05}, we use cores with $k$ in 10, 50, 100
and 1000. In the smaller graphs we use cores with $k$ in 10, 30, 50 and
100. \texttt{twitter-2010} and \texttt{com-orkut} only had few cores where the
minimum cut is not trivial. Therefore we used those cores. As
\texttt{hollywood-2011} is very dense, we multiplied the $k$ value of all cores by a factor of 2.

\label{cap}

\begin{algorithm*}[ht!]
  \begin{algorithmic}[1]
    \INPUT $G = (V,E,c) \leftarrow$ undirected graph with integer edge weights, $\hat\lambda \leftarrow$ upper bound for minimum cut
    \OUTPUT $\mathcal{T} \leftarrow$ forest of contractible edges
    \State Label all vertices $v \in V$ ``unvisited'', blacklist $\mathcal{B}$ empty
    \State $\forall v \in V: r(v) \leftarrow 0$
    \State $\forall e \in E: q(e) \leftarrow 0$
    \State $\mathcal{Q} \leftarrow$ empty priority queue
    \State Insert random vertex into $\mathcal{Q}$
    \While{$\mathcal{Q}$ not empty}
    \State $x \leftarrow \mathcal{Q}$.pop\_max() \Comment{Choose unvisited vertex with highest priority}

    \State $\alpha \leftarrow \alpha + c(x) - 2 r(x)$

    \State $\hat\lambda \leftarrow min(\hat\lambda, \alpha)$
    \For{$e = (x,y) \leftarrow$ edge to vertex $y$ not in $\mathcal{B}$ and not visited}
    \If{$r(y) < \hat\lambda \leq r(y) + c(e)$}
    \State $T \leftarrow T \cup e$ \Comment{Mark edge $e$ to contract}
    \EndIf
    \State $r(y) \leftarrow r(y) + c(e)$
    \State $q(e) \leftarrow r(y)$
    \State $\mathcal{Q}(y) \leftarrow r(y)$

    \EndFor

    \State Mark $x$ ``visited''
    \EndWhile

  \end{algorithmic}
  \caption{\label{algo:cap} CAPFOREST}
\end{algorithm*}

\end{appendix}

\end{document}